\documentclass{article}

\usepackage[english]{babel}

\usepackage[a4paper,top=2cm,bottom=2cm,left=3cm,right=3cm,marginparwidth=1.75cm]{geometry}

\usepackage[colorlinks=true, allcolors=blue]{hyperref}
\usepackage{times}

\usepackage{soul}
\usepackage{url}
\usepackage[utf8]{inputenc}
\usepackage[small]{caption}
\usepackage{graphicx}
\usepackage{amsmath,amssymb,amsthm}
\usepackage{amsfonts}
\usepackage{booktabs}
\usepackage{color}
\usepackage{xfrac}
\usepackage{cleveref}
\usepackage{dsfont}
\usepackage{pgfplots}						

\usepackage[ruled,linesnumbered,vlined]{algorithm2e}
\usepackage{tikz}
\usepackage{pgfplots}
\pgfplotsset{compat=1.17}
\usetikzlibrary{patterns}
\usepackage{pgfplotstable}
    \pgfplotsset{
    name nodes near coords/.style={
        every node near coord/.append style={
            name=#1-\coordindex,
            alias=#1-last,
        },
    },
    name nodes near coords/.default=coordnode
    }

\newtheorem{theorem}{Theorem}[section]

\newtheorem{lemma}[theorem]{Lemma}
\newtheorem{observation}[theorem]{Observation}
\newtheorem{example}[theorem]{Example}
\newtheorem{definition}[theorem]{Definition}

\makeatletter
\newtheorem*{rep@theorem}{\rep@title}
\newcommand{\newreptheorem}[2]{%
\newenvironment{rep#1}[1]{%
 \def\rep@title{#2 \ref{##1}}%
 \begin{rep@theorem}}%
 {\end{rep@theorem}}}
\makeatother

\newreptheorem{lemma}{Lemma}
\newreptheorem{claim}{Claim}

\newcommand{\bR}{{\mathbb{R}}}

\newcommand{\PROP}{\mathsf{PROP}}
\newcommand{\WPROP}{\mathsf{WPROP}}
\newcommand{\WPROPX}{\mathsf{WPROPX}}

\newcommand{\bX}{\mathbf{X}}
\newcommand{\bs}{\mathbf{s}}
\newcommand{\bc}{\mathbf{c}}
\newcommand{\bw}{\mathbf{w}}

\DeclareMathOperator*{\argmax}{argmax}
\DeclareMathOperator*{\argmin}{argmin}

\title{One Quarter Each (on Average) Ensures Proportionality}

\author{
Xiaowei Wu
\thanks{IOTSC, University of Macau. \{xiaoweiwu,yc27429,yc17423\}@um.edu.mo. The authors are ordered alphabetically.}
\and Cong Zhang $^*$
\and Shengwei Zhou $^*$
}

\date{}

\begin{document}
\pagestyle{plain}
\maketitle

\begin{abstract}
We consider the problem of fair allocation of $m$ indivisible items to a group of $n$ agents with subsidy (money).
Our work mainly focuses on the allocation of chores but most of our results extend to the allocation of goods as well.
We consider the case when agents have (general) additive cost functions.
Assuming that the maximum cost of an item to an agent can be compensated by one dollar, we show that a total of $n/4$ dollars of subsidy suffices to ensure a proportional allocation.
Moreover, we show that $n/4$ is tight in the sense that there exists an instance with $n$ agents for which every proportional allocation requires a total subsidy of at least $n/4$. 
We also consider the weighted case and show that a total subsidy of $(n-1)/2$ suffices to ensure a weighted proportional allocation. 
\end{abstract}
	
\section{Introduction} \label{sec:intro}
	
We consider the problem of fairly allocating a set of $m$ indivisible items $M$ to a group of $n$ heterogeneous agents $N$, where an allocation $\bX$ is a partition of the items $M$ into $n$ disjoint bundles, each of which is allocated to a unique agent.
In this paper, we consider both the allocation of goods and chores.
When the items are goods, we use $v_i: 2^M \to \mathbb{R}^+\cup \{0\}$ to denote the \emph{valuation function} of agent $i$, and we say that agent $i$ has \emph{value} $v_i(S)$ for the bundle of items $S\subseteq M$.
When items are chores, we use $c_i: 2^M \to \mathbb{R}^+\cup \{0\}$ to denote the \emph{cost function} of agent $i$, and say that agent $i$ has \emph{cost} $c_i(S)$ for bundle $S$.
We assume that the valuation/cost functions are additive.
Among the different notions to measure the fairness of an allocation, envy-freeness~\cite{foley1967resource,conf/sigecom/LiptonMMS04,journals/teco/CaragiannisKMPS19} and proportionality~\cite{steihaus1948problem} are arguably the most well-studied two.
An allocation $\bX = (X_1,\ldots,X_n)$ if called \emph{envy-free} (EF) if every agent weakly prefers her own bundle than any other bundle, i.e., $v_i(X_i) \geq v_i(X_j)$ for all $i,j\in N$ (for goods) or $c_i(X_i) \leq c_i(X_j)$ for all $i,j\in N$ (for chores).
The allocation is called \emph{proportional} (PROP) if every agent receives a bundle at least as good as her proportional share, i.e., $v_i(X_i) \geq v_i(M)/n$ for all $i\in N$ (for goods) or $c_i(X_i) \leq c_i(M)/n$ for all $i\in N$ (for chores).
Clearly, every EF allocation is PROP.
Unfortunately, when items are indivisible, EF/PROP allocations are not guaranteed to exist, e.g., consider allocating a single item to two agents.
Existing works have taken two different paths to circumvent this non-existence result, one by considering relaxations of the fairness notions, and the other by introducing money (a divisible good) to eliminate the inevitable unfairness.
	
\paragraph{Relaxations.}
Various relaxations of EF have been proposed in the past decades, among which EF1 and EFX are the most popular ones.
The concept of {\em envy-freeness up to one item} (EF1) is proposed by Budish~\cite{conf/bqgt/Budish10}, and requires that the envy between any two agents can be eliminated by removing a single item.
The notion of {\em envy-freeness up to any item} (EFX) is proposed by Caragiannis et al.~\cite{journals/teco/CaragiannisKMPS19} and is defined in a similar manner, but requires that the envy can be eliminated by removing any item from the envied agent (for goods) or the envious agent (for chores).
EF1 allocations are guaranteed to exist and can be efficiently computed for goods~\cite{conf/sigecom/LiptonMMS04}, chores, and even mixture of goods and chores~\cite{conf/approx/BhaskarSV21,journals/aamas/AzizCIW22}. However, unlike EF1, EFX allocations are known to exists only for some special cases, e.g., see \cite{journals/siamdm/PlautR20,conf/sigecom/ChaudhuryGM20,journals/tcs/AmanatidisBFHV21} for goods and \cite{conf/ijcai/0002022,conf/www/0037L022,journals/corr/abs-2211-00879,journals/corr/abs-2211-15836} for chores.
Whether EFX allocations always exist remains the most interesting open problem in this research area.
Similar to EF1 and EFX, we can relax PROP to \emph{proportionality up to one item} (PROP1) and \emph{proportionality up to any item} (PROPX).
Like EF1, PROP1 allocations always exist and can be efficiently computed for goods~\cite{conf/sigecom/ConitzerF017,conf/aaai/BarmanK19}, chores~\cite{journals/corr/abs-1907-01766} and mixture of goods and chores~\cite{journals/orl/AzizMS20}.
For the allocation of goods, Aziz et al.~\cite{journals/orl/AzizMS20} show that PROPX allocations may not exist.
In contrast, PROPX allocations for chores always exist and can be efficiently computed~\cite{moulin2018fair,conf/www/0037L022}. 
For a more comprehensive review of the existing works, please refer to the recent surveys~\cite{AMANATIDIS2023103965,journals/sigecom/AzizLMW22}.

\paragraph{Fair Allocation with Money.}
Since unfairness is inevitable, a natural idea is to compensate some agents with a subsidy to eliminate envy or achieve proportionality.
Specifically, suppose that each agent $i\in N$ receives a subsidy $s_i\geq 0$, we say that the resulting allocation (with subsidy) is EF if for all $i,j\in N$, $v_i(X_i) + s_i \geq v_i(X_j) + s_j$ (for goods); or $c_i(X_i) - s_i \leq c_i(X_j) - s_j$ (for chores).
We are interested in computing an allocation with a small amount of total subsidy to achieve envy-freeness, assuming that each item has value/cost at most one to each agent.
For the allocation of goods, Halpern and Shah~\cite{conf/sagt/HalpernS19} show that a total subsidy of $m(n-1)$ dollars suffices.
The result was then improved to $n-1$ dollars by Brustle et al.~\cite{conf/sigecom/BrustleDNSV20}, who also showed that it suffices to subsidize each agent at most one dollar and that the allocation without subsidy is EF1.
Note that $n-1$ dollars are necessary to guarantee envy-freeness: consider allocating a single good with value $1$ to $n$ identical agents. 
Since every EF allocation is PROP, the result of Brustle et al. also holds for achieving proportionality.
However, it remains unknown whether $n-1$ dollars are necessary to ensure proportionality.
As far as we know, the problem of computing EF or PROP allocations with subsidy has not been considered for the allocation of chores.

\subsection{Our Results} \label{ssec:results}
		
Our work aims at filling in the gaps for the fair allocation problem with subsidy, for both goods and chores.
Ideally, we would like to compute allocations that are EF or PROP with a small amount of subsidy, and also satisfy some relaxations of EF and PROP without the subsidy.
We mainly focus on the allocation of chores, but most of our results extend to the allocation of goods.
We use $\bX = (X_1,\ldots,X_n)$ to denote an allocation, $\bs = (s_1,\ldots,s_n)\in [0,1]^n$ to denote the subsidies and $\|\bs\|_1 = \sum_{i\in N} s_i$ to denote the total subsidy.
We first show that the result of Brustle et al.~\cite{conf/sigecom/BrustleDNSV20} (for computing EF allocation for goods) can be straightforwardly extended to the allocation of chores: it suffices to compensate each agent a subsidy at most one dollar to achieve an EF allocation with a total subsidy at most $n-1$.
	
\smallskip
\noindent
\textbf{Result 1.}
For the allocation of chores, we can compute in polynomial time an EF1 allocation $\bX$ and subsidies $\bs$ such that $(\bX,\bs)$ is EF, where the total subsidy $\|\bs\|_1 \leq n-1$.
Moreover, there exists an instance for which every EF allocation requires a total subsidy of at least $n-1$.
\smallskip
	
The main results of this paper concern the fairness notion of proportionality.
Note that unlike envy-freeness, proportionality can be achieved given any allocation and sufficient subsidy, e.g., by setting $s_i = \max\{c_i(X_i) - c_i(M)/n,0\}$ for all $i\in N$.
Therefore the interesting question is how much subsidy is sufficient to guarantee the existence of PROP allocations.
Since every EF allocation is PROP, this amount is at most $n-1$.
However, whether strictly less subsidy is sufficient remains unknown.
In this work, we answer this question by showing that a total subsidy of $n/4$ suffices to guarantee the existence of PROP allocation.
We propose polynomial-time algorithms for computing such allocations and subsidies.
Moreover, we show that the computed allocations satisfy PROPX when agents have identical additive cost functions, and PROP1 when agents have general additive cost functions.

\smallskip
\noindent
\textbf{Result 2.}
For the allocation of chores to a group of $n$ agents with identical cost functions, we can compute in polynomial time a PROPX allocation $\bX$ and subsidies $\bs$ such that $(\bX,\bs)$ is PROP, where the total subsidy $\|\bs\|_1 \leq n/4$.
Moreover, there exists an instance with $n$ identical agents for which every PROP allocation requires a total subsidy of at least $n/4$.
The same results also hold for weighted agents. 
\smallskip
	
\noindent
\textbf{Result 3.}
For the allocation of chores to a group of $n$ agents with general additive cost functions, we can compute in polynomial time a PROP1 allocation $\bX$ and subsidies $\bs$ such that $(\bX,\bs)$ is PROP, where the total subsidy $\|\bs\|_1 \leq n/4$.
When agents have arbitrary weights, the total subsidy is at most $(n-1)/2$.
\smallskip
	
The above set of results (Result 3) extends straightforwardly to the allocation of goods (whose proof can be found in the appendix).
Result 2 does not extend to the allocation of goods since PROPX allocations are not guaranteed to exist for goods.
	
Our results settle the problem of characterizing the subsidy required for ensuring proportionality, which is overlooked in existing works.
The results indicate that for the fair allocation problem with subsidy, proportionality is strictly cheaper to achieve, compared with envy-freeness. 
Our results for general additive cost functions are achieved by rounding a well-structured fractional allocation to an integral allocation with subsidy, which is novel and might be useful to solve other research problems in the area of fair allocation.

\subsection{Other Related Works}
	

A similar setting introduced by Aziz~\cite{conf/aaai/000121} is the fair allocation with monetary transfers, which allows agents to transfer money to each other.
Different from our objective of minimizing the total subsidy (money) to achieve fairness, the main result of the paper is to provide a characterization of allocations that are equitable and envy-free with monetary transfers.
Beyond additive valuation functions, Brustle et al.~\cite{conf/sigecom/BrustleDNSV20} show that an envy-free allocation for goods always exists with a subsidy of at most $2(n-1)$ dollars per agent under general monotonic valuation functions.
Barman et al.~\cite{journals/corr/abs-2201-07419} consider the dichotomous valuations, i.e., the marginal value for any good to any agent is either $0$ or $1$.
They show that there exists an allocation that achieves envy-freeness with a per-agent subsidy of at most $1$.
Goko et al.~\cite{conf/atal/GokoIKMSTYY22} consider the fair and truthful mechanism with limited subsidy.
They show that under general monotone submodular valuations, there exists a truthful allocation mechanism that achieves envy-freeness and utilitarian optimality by subsidizing each agent at most $1$ dollar.

Another closely related work is the rent division problem that focuses on allocating $m=n$ indivisible goods among $n$ agents and dividing the fixed rent among the agents~\cite{aragones1995derivation,edward1999rental,journals/jacm/GalMPZ17,klijn2000algorithm,maskin1987fair}.
More general models where envy-freeness is achieved with money have also been considered in~\cite{haake2002bidding,meertens2002envy}. 
Recently, Peters et al.~\cite{conf/nips/PetersPZ22} study the robustness of the rent division problem where each agent may misreport her valuation of the rooms. 

When considering money as a divisible good, our setting is similar to the fair allocation of mixed divisible and indivisible items.
Bei et al.~\cite{journals/ai/BeiLLLL21} study the problem of fairly allocating a set of resources containing both divisible and indivisible goods.
They propose the fairness notion of envy-freeness for mixed goods (EFM), and prove that EFM allocations always exist for agents with additive valuations. 
Bhaskar et al~\cite{conf/approx/BhaskarSV21} show that envy-free allocations always exist for mixed resources consisting of doubly-monotonic indivisible items and a divisible chore, which completes the result of Bei et al.~\cite{journals/ai/BeiLLLL21}.
Recently, Li et al.~\cite{journals/corr/abs-2305-09206} consider the truthfulness of EFM allocation.
They show that EFM and truthful allocation mechanisms do not exist in general and design truthful EFM mechanisms for several special cases. 
For a more detailed review for the existing works on mixed fair allocation, please refer to the recent survey~\cite{journals/corr/abs-2306-09564}.

\subsection{Organization of the Paper}
	
Since the result for computing EF allocation with subsidy for chores is a straightforward extension of the result by Brustle et al.~\cite{conf/sigecom/BrustleDNSV20}, we defer the proofs to Appendix~\ref{sec:ef-chores}.
We first provide the notations and definitions (for the allocation of chores) in Section~\ref{sec:prelim}.
We prove Result 2 (the unweighted case) in Section~\ref{sec:identical}.
The extension to the weighted setting is included in Section~\ref{ssec:weighted_identical}.
We prove Result 3 (the unweighted case) in Section~\ref{sec:general}.
The extension to the weighted setting is included in Section~\ref{ssec:weighted_general}.
The extensions of the above results to the allocation of goods are also deferred to the appendix (see Appendix~\ref{sec:goods}).
We conclude the paper and discuss the open questions in Section~\ref{sec:conclusion}.

\section{Preliminary} \label{sec:prelim}

In the following, we introduce the notations and the fairness notions for the allocation chores.
Those for goods will be introduced in Appendix~\ref{sec:goods}.
We assume that the agents are unweighted. 
The weighted setting will be considered in Section~\ref{sec:weighted}.
We consider the problem of allocating $m$ indivisible chores $M$ to $n$ agents $N$ where each agent $i\in N$ has an additive cost function $c_i:2^M \to \bR^+ \cup \{0\}$.
A cost function $c_i$ is said to be {\em additive} if for any bundle $S \subseteq M$ we have $c_i(S) = \sum_{e\in S} c_i(\{e\})$.
For convenience, we use $c_i(e)$ to denote $c_i(\{e\})$.
We use $\bc = (c_1, . . . , c_n)$ to denote the cost functions of agents.
We assume w.l.o.g. that each item has cost at most one to each agent, i.e. $c_i(e) \leq 1$ for any $i\in N$, $e\in M$.
An allocation is represented by an $n$-partition $\bX = (X_1,\ldots,X_n)$ of the items, where $X_i \cap X_j = \emptyset$ for all $i \neq j$ and $\cup_{i\in N} X_i = M$.
In allocation $\bX$, agent $i\in N$ receives bundle $X_i$.
For convenience of notation, given any set $X\subseteq M$ and $e\in M$, we use $X+e$ and $X-e$ to denote $X\cup\{e\}$ and $X\setminus\{e\}$, respectively.

\begin{definition}[PROP]
    An allocation $\bX$ is called proportional (PROP) if $c_i(X_i) \leq \frac{c_i(M)}{n}$ for all $i\in N$.
\end{definition}

We use $\PROP_i$ to denote agent $i$'s proportional share, i.e., $\PROP_i = \frac{c_i(M)}{n}$.

\begin{definition}[PROP1]
    An allocation $\bX$ is called proportional up to one item (PROP1) if for any $i\in N$, there exists an item $e\in X_i$ such that $c_i(X_i - e) \leq \PROP_i$.
\end{definition}

\begin{definition}[PROPX]
    An allocation $\bX$ is called proportional up to any item (PROPX) if for any $i\in N$, and any item $e\in X_i$, we have $c_i(X_i - e) \leq \PROP_i$.
\end{definition}

We use $s_i \geq 0$ to denote the subsidy we give to agent $i\in N$, $\bs = (s_1, \ldots, s_n)$ to denote the set of subsidies, and $\|\bs\|_1 = \sum_{i\in N} s_i$ to denote the total subsidy.

\begin{definition}[PROPS]
    An allocation $\bX$ with subsidies $\bs = (s_1, \ldots, s_n)$ is called proportional with subsidies (PROPS) if for any $i\in N$ we have $c_i(X_i) - s_i \leq \PROP_i$.
\end{definition}

Given any instance, we aim to find PROPS allocation $\bX$ with a small amount of total subsidy.
Unlike envy-freeness with subsidy, given any allocation $\bX$, computing the minimum subsidy to achieve proportionality can be trivially done by setting
\begin{equation*}
    s_i = \max \{c_i(X_i) - \PROP_i, 0\}, \qquad \forall i\in N.
\end{equation*}

Therefore, in the rest of this paper, we mainly focus on computing the allocation $\bX$. The subsidy to each agent will be automatically decided by the above equation.

\section{Identical Cost Functions} \label{sec:identical}

In this section, we focus on the computation of PROPS allocations when agents have identical cost functions, i.e., $c_i(\cdot) = c(\cdot)$ for all agents $i\in N$.
We use $\PROP$ to denote the proportional share of all agents, i.e., $\PROP = c(M)/n$.
Before we present our algorithmic results, we first show a lower bound on the total subsidy required to achieve proportionality.

\begin{lemma}\label{lemma:lower-bounds}
    Given any $n\geq 2$, there exists an instance with $n$ agents for which every PROPS allocation requires a total subsidy of at least $n/4$ (when $n$ is even); at least $(n^2-1)/(4n)$ (when $n$ is odd).
\end{lemma}
\begin{proof}
    Suppose $n\geq 2$ is even.
    Consider the instance with $n$ agents and $n/2$ items where each item has cost $1$ to all agents.
    For every agent $i\in N$, her proportional share is $\PROP = 1/2$.
    Consider any allocation $\bX$, and suppose that $k \leq n/2$ agents receive at least one item.
    Then each of these agents $i$ requires a subsidy of $c(X_i) - 1/2$, which implies $\|\bs\|_1 = c(M) - k/2 \geq n/4$.
    In other words, any PROPS allocation requires a total subsidy of at least $n/4$.

    Suppose $n\geq 2$ is odd.
    Consider the instance with $n$ agents and $(n-1)/2$ items where each item has cost $1$ to all agents. 
    For every agent $i\in N$, her proportional share is $\PROP = (n-1)/(2n)$.
    Following a similar analysis as above we can show that the total subsidy required by any PROPS allocation is at least
    \begin{equation*}
        c(M) - \frac{n-1}{2} \cdot \frac{n-1}{2n} = \frac{n-1}{2}\cdot \frac{n+1}{2n} = \frac{n^2 - 1}{4n}.
    \end{equation*}
    
    Therefore, every PROPS allocation requires a total subsidy at least $(n^2-1)/(4n)$.
\end{proof}

Next, we present an algorithm for computing PROPS allocations that require a total subsidy matching the above lower bound.
We use the Load Balancing Algorithm to compute an allocation $\bX$ that is PROPX (for $n$ identical agents), and show that the total subsidy required to achieve proportionality is at most $n/4$ (when $n$ is even); at most $(n^2-1)/(4n)$ (when $n$ is odd).
By re-indexing the items, we can assume w.l.o.g. that the items are sorted in decreasing order of costs, i.e. $c(e_1) \geq c(e_2) \geq \cdots \geq c(e_m)$.
During the algorithm, we greedily allocate items $e_1,e_2,\ldots,e_m$ one-by-one to the agent with minimum bundle cost, i.e., $\argmin_{i\in N} c(X_i)$.
We summarize the steps of the full algorithm in Algorithm~\ref{alg:LBA}.

\begin{algorithm}[htbp]
    \caption{Load Balancing Algorithm}
    \label{alg:LBA}
    \KwIn{An instance $(M,N,\bc)$ with identical agents and $c(e_1) \geq c(e_2) \geq \cdots \geq c(e_m)$.}
    For all $i\in N$, let $X_i \gets \emptyset$ \;
    \For{$j = 1,2,\dots, m$}{
        Let $i^* \gets \argmin_{i\in N} c(X_i)$\;
        Update $X_{i^*} \gets X_{i^*} \cup \{e_j\}$\;
    }
    \KwOut{An allocation $\bX = \{X_1,\dots,X_n\}$.}
\end{algorithm}

\begin{lemma}\label{lemma:identical-propx}
    The Load Balancing Algorithm (Algorithm~\ref{alg:LBA}) computes a PROPX allocation given any instance with $n$ agents having identical cost functions.
\end{lemma}
\begin{proof}
    Fix any agent $i\in N$ and let $e_{\sigma(i)}$ be the last item agent $i$ receives, it suffices to show that $c(X_i - e_{\sigma(i)}) \leq \PROP$ since items are allocated in the order of descending costs.
    Assume otherwise, i.e., $c(X_i - e_{\sigma(i)}) > \PROP$.
    Then at the moment when item $e_{\sigma(i)}$ was allocated, we have $c(X_j) \geq c(X_i - e_{\sigma(i)}) > \PROP$ for all $j\neq i$, which leads to a contradiction that $c(M) = \sum_{i\in N} c(X_i) > n \cdot \PROP = c(M)$.
\end{proof}

Next, we provide upper bounds on the total subsidy required to make $\bX$ a PROPS allocation, which exactly match the lower bounds given in Lemma~\ref{lemma:lower-bounds}.

\begin{theorem}\label{theorem:identical}
    Given any instance with $n$ agents having identical cost functions, there exists a PROPS allocation with total subsidy at most $n/4$ when $n$ is even, and at most $(n^2-1)/(4n)$ when $n$ is odd.
\end{theorem}
\begin{proof}
    Given the allocation $\bX$ returned by Algorithm~\ref{alg:LBA}, we first partition the agents into two disjoint groups $N_1$ and $N_2$ depending on whether the allocation is proportional to her as follows.
    \begin{equation*}
        N_1 = \{i\in N : c(X_i) > \PROP\}, \quad N_2 = \{i\in N : c(X_i)\le \PROP\}.
    \end{equation*}

    For all $i\in N_1$, we use $e_{\sigma(i)}$ to denote the last item allocated to agent $i$.
    Then we define
    \begin{equation*}
        h_i = 
        \begin{cases}
            \PROP - c(X_i - e_{\sigma(i)}), \quad &\forall i \in N_1 \\
            \PROP - c(X_i), \quad & \forall i \in N_2
        \end{cases}
        .
    \end{equation*}

    Note that by Lemma~\ref{lemma:identical-propx}, we have $h_i \geq 0$ for all $i\in N$. Moreover, since we only need to subsidize agents in $N_1$, the total subsidy required to achieve proportionality can be expressed as:
    \begin{align*}
        \|\bs\|_1 & = \sum_{i\in N_1} (c(X_i) - \PROP) = c(M) - \sum_{i\in N_2} c(X_i) - \frac{|N_1|}{n}\cdot c(M) \\
        & = \frac{|N_2|}{n}\cdot c(M) - \sum_{i\in N_2} c(X_i) = \sum_{i\in N_2} (\PROP - c(X_i)) = \sum_{i\in N_2} h_i.
    \end{align*}

    On the other hand, we show that the sequence $(h_1,\ldots,h_n)$ has some useful properties.
    
    First, using the same argument as we have shown in the proof of Lemma~\ref{lemma:identical-propx}, for all $i\in N_1$ and $j\in N_2$, we have $c(X_i - e_{\sigma(i)}) \leq c(X_j)$ (because otherwise item $e_{\sigma(i)}$ will not be allocated to agent $i$).
    Hence we have $h_i \geq h_j$ for all $i\in N_1$ and $j\in N_2$.
    By renaming the agents, we can assume w.l.o.g. that $h_1 \geq h_2 \geq \cdots \geq h_n$ (see Figure~\ref{fig:ident_val} for an illustrating example).
    
    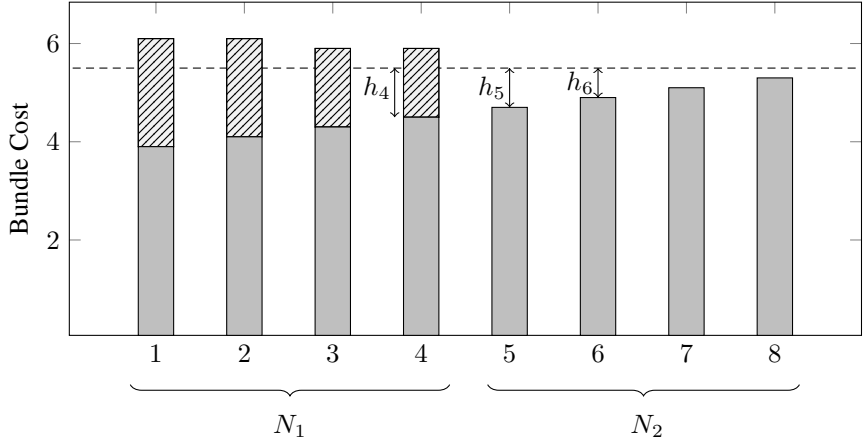
\begin{figure}[htbp]
    \centering
    \begin{tikzpicture}						
        \begin{axis}[
        ybar stacked,
        enlargelimits=0.14,
        width = 0.8\textwidth,
	  height = 0.4\textwidth,
        bar width=0.4, 
        bar shift=0pt,
        xtick=data,
        ymin=0.8,
        legend style={draw=none},
        legend pos = north west,
	ylabel =  Bundle Cost,
        ] 
        
        	\addplot[draw=black, fill=gray!50] coordinates{
        		(1,3.9)
                    (2,4.1)
                    (3,4.3)
                    (4,4.5)
                    (5,4.7)
                    (6,4.9)
                    (7,5.1)
                    (8,5.3)
        	};
        	\addplot[draw=black,postaction={pattern=north east lines}, fill=gray!10] coordinates{
                    
                    (1,2.2)
                    (2,2.0)
                    (3,1.6)
                    (4,1.4)
        	};
         
        \draw [<->] (5,5.5)--(5,4.7);
        \node at (4.8,5.1) {$h_{5}$};
        \draw [<->] (6,5.5)--(6,4.9);
        \node at (5.8,5.2) {$h_6$};
        
        \draw [<->] (3.7,5.5)--(3.7,4.5);
        \node at (3.5,5.1) {$h_4$};
        \draw [densely dashed] (-1,5.5)--(10,5.5); 
        \end{axis}
        \draw[decorate,decoration={brace,raise=10pt,amplitude=0.15cm},black] (5,-0.3)--(0.8,-0.3);
        \node at (2.9,-1.2) {$N_1$};
        \draw[decorate,decoration={brace,raise=10pt,amplitude=0.15cm},black] (9.6,-0.3)--(5.5,-0.3);
        \node at (7.6,-1.2) {$N_2$};

    \end{tikzpicture}
    \caption{An illustrating example for the allocation $\bX$ with agent groups $N_1$ and $N_2$. The dashed areas represent the last items received by agents in $N_1$. The horizontal dashed line represents the proportional share for all the agents.}
    \label{fig:ident_val}
    \end{figure}

    Second, by definition we have
    \begin{align*}
        \sum_{i \in N} h_i & = \sum_{i \in N_1} (\PROP-c(X_i)+c(e_{\sigma(i)}))+\sum_{i \in N_2} (\PROP-c(X_i)) \\
        & = n\cdot \PROP - \sum_{i\in N} c(X_i) + \sum_{i \in N_1} c(e_{\sigma(i)})
        = \sum_{i \in N_1} c(e_{\sigma(i)}) \leq |N_1|,
    \end{align*}
    where in the last inequality we use the assumption that each item cost at most one to each agent.

    Making use of the two properties, we are now ready to prove the theorem: the total subsidy required is
    \begin{align*}
        \sum_{i\in N_2} h_i \leq \frac{|N_2|}{n}\cdot \sum_{i\in N} h_i \leq \frac{|N_1|\cdot |N_2|}{n}\le 
        \begin{cases}
            n/4, \quad &\text{when $n$ is even}\\
            (n^2-1)/(4n),  \quad &\text{when $n$ is odd}
        \end{cases},
    \end{align*}
    where in the first inequality we use the property that $h_i \geq h_j$ for all $i\in N_1$ and $j\in N_2$.    
\end{proof}

\section{Agents with General Additive Cost Functions} \label{sec:general}

In this section, we focus on the case when agents have general additive cost functions, i.e., each agent $i$ has cost function $c_i$, and show that a total subsidy of $n/4$ suffices to ensure proportionality.
We first present a reduction showing that if we have an algorithm for computing PROPS allocations for identical ordering instances, we can convert it to an algorithm that works for general instances while preserving the subsidy requirement.
Similar reductions are widely used in the computation of approximate MMS allocations~\cite{journals/teco/BarmanK20,journals/aamas/BouveretL16,conf/sigecom/HuangL21,journals/corr/abs-2302-04581} and PROPX allocations~\cite{conf/www/0037L022}.

\begin{definition}[Identical Ordering (IDO) Instances]
    An instance is called identical ordering (IDO) if all agents have the same ordinal preference on the items, i.e., $c_i(e_1) \geq c_i(e_2) \geq \cdots \geq c_i(e_m)$ for all $i\in N$.
\end{definition}

\begin{lemma} \label{lemma:reduction-to-IDO}
    If there exists a polynomial time algorithm that given any IDO instance computes a PROPS allocation with total subsidy at most $\alpha$, then there exists a polynomial time algorithm that given any instance computes a PROPS allocation with total subsidy at most $\alpha$.
\end{lemma}
\begin{proof}
    Given any instance $\mathcal{I} = (M,N,\bc)$, we construct an IDO instance $\mathcal{I'} = (M,N,\bc')$ where $\bc' = (c'_1, \ldots, c'_n)$ is defined as follows.
    Let $\sigma_i(k)\in M$ be the $k$-th most costly item under cost function $c_i$.
    Let $c'_i(e_k) = c_i(\sigma_i(k))$.
    Thus with cost function $\bc'$, the instances $\mathcal{I'}$ is IDO.
    Note that for all $i\in N$ we have $c'_i(M) = c_i(M)$.
    Then we run the algorithm for the IDO instance $\mathcal{I'}$ and get a PROPS allocation $\bX'$ with subsidy $\bs'$ such that $\|\bs'\|_1 \leq \alpha$.
    By definition, for all agent $i\in N$ we have
    \begin{equation*}
        c'_i(X'_i) - s'_i \leq \frac{1}{n} \cdot c'_i(M) = \frac{1}{n}\cdot c_i(M).
    \end{equation*}
    
    In the following, we use $\bX'$ to guide us on computing a PROPS allocation $\bX$ with $\bs$ for instance $\mathcal{I}$.
    We show that $c_i(X_i) \leq c'_i(X'_i)$ for all $i\in N$, which implies $\|\bs\|_1 \leq \|\bs'\|_1 \leq \alpha$.

    We initialize $X_i = \emptyset$ for all $i\in N$ and let $P = M$ be the set of unallocated items.
    Sequentially for $j = m,m-1,\ldots,1$, we let the agent $i$ who receives item $e_j$ under allocation $\bX'$, i.e., $e_j \in X'_i$, pick her favorite unallocated item, i.e., update $X_i \gets X_i + e$ and $P\gets P-e$ for $e = \argmin_{e'\in P}\{c_i(e')\}$.
    At the beginning of each round $j$, we have $|P| = j$.
    Since $e_j$ is the $j$-th most costly item under cost function $c'_i$, we must have $c_i(e) \leq c'_i(e_j)$ for the item $e$ agent $i$ picks during round $j$.
    Therefore we can establish a one-to-one correspondence between items in $X_i$ and $X'_i$ satisfying the above inequality, which implies $c_i(X_i) \leq c'_i(X'_i)$.
\end{proof}

With the above reduction, in the following, we only consider IDO instances.
Our algorithm has two main steps: we first compute a fractional PROP allocation, in which a small number of items are fractionally allocated; then we find a way to round the fractional allocation to an integral one.
Since some agent may have cost exceeding her proportional share after rounding, we offer subsidies to these agents.
By carefully deciding the rounding scheme, we show that the total subsidy required is at most $n/4$.

\subsection{Computing an Allocation with at most \texorpdfstring{$n-1$}{} Factional Items}
\label{ssec:exactly-n-1}

In this section, we use the classic Moving Knife Algorithm to compute a fractional PROP allocation, based on which we compute the PROPS allocation $\bX$ with subsidy $\bs$.

\paragraph{The Algorithm.}
For ease of discussion, we interpret the $m$ items as an interval $(0,m]$, where item $e_i$ corresponds to interval $(i-1, i]$.
We interpret every interval as a bundle of items, where some items might be fractional.
Specifically, interval $(l,r]$ contains $(\lceil l \rceil - l)$-fraction of item $e_{\lceil l \rceil}$, $(r - \lfloor r \rfloor)$-fraction of item $e_{\lceil r \rceil}$ and integral item $e_j$ for every integer $j$ satisfying $(j-1,j] \subseteq (l,r]$.
The cost of the interval to each agent $i\in N$ is also defined in the natural way:
\begin{equation*}
    c_i(l,r) = (\lceil l \rceil - l)\cdot c_i(e_{\lceil l \rceil}) + \sum_{j = \lceil l \rceil + 1}^{\lfloor r \rfloor} c_i(e_j) + (r - \lfloor r \rfloor)\cdot c_i(e_{\lceil r \rceil}).
\end{equation*}
The algorithm proceeds in rounds, where in each round some agent picks an interval and leaves.
We maintain that at the beginning of each round, the remaining set of items forms a continuous interval $(l,m]$.
In each round, we imagine that there is a moving knife that moves from the leftmost position $l$ to the right.
Each agent shouts if she thinks that the cost of an interval passed by the knife is equal to her proportional share.
The last agent\footnote{In this paper we break tie arbitrarily but consistently, e.g., by agent id.} who shouts picks the interval passed by the knife and leaves, and the algorithm recurs on the remaining interval.
If in some round the knife reaches the end of the interval, any agent who has not shouted picks the whole interval $(l,m]$ and the algorithm terminates.
The steps of the full algorithm are summarized in Algorithm~\ref{alg:MKA}.



\begin{algorithm}[htbp]
    \caption{The Moving Knife Algorithm}
    \label{alg:MKA}
    \KwIn{The interval $(0,m]$ corresponding to all items $M$, agents $N$, cost function $\bc = (c_1, \ldots, c_n)$.}
    Initialize $X^0_i \gets \emptyset$ for each $i \in N$, and $l \gets 0$\;
    \While{$l \neq m$}{
        Let $r_i \gets \max \{r \leq m: c_i(l,r) \leq c_i(M)/n\}$ for all $i \in N$\;
        Let $i^* \gets \argmax\{r_i\}$\;
        Update $X^0_{i^*} \gets (l,r_{i^*}]$\;
        Update $N \gets N\setminus \{i^*\}$, $l \gets r_{i^*}$\;
    }
    \KwOut{Fractional allocation $\bX^0 = (X^0_1, \ldots, X^0_n)
    $.}
\end{algorithm}

It has been shown by Aziz et al.~\cite{journals/aamas/AzizCIW22} that the Moving Knife Algorithm computes a PROP allocation for divisible chores.
Therefore we have the following lemma immediately.
For completeness, we give a short proof.
By renaming the agents, we can assume w.l.o.g. that agents are indexed by their picking order, i.e., agent $i$ is the $i$-th agent who picks and leaves.

\begin{lemma}\label{lemma:MKA}
     The Moving Knife Algorithm computes fractional PROP allocations in polynomial time.
\end{lemma}
\begin{proof}
    Let $\bX^0 = (X^0_1, \ldots, X^0_n)$ be the allocation returned by the algorithm.
    Consider agent $n$ who receives the last bundle.
    If $X^0_n$ is empty then the allocation is clearly PROP to agent $n$.
    Suppose $X^0_n$ is not empty.
    For any other agent $i \neq n$ we must have $c_n(X^0_{i}) \geq \frac{1}{n} \cdot c_n(M)$ since at the round when agent $i$ picked a bundle, agent $i$ shouts not earlier than agent $n$.
    Hence we have
    \begin{equation*}
        c_n(X_n) = c_n(M) - \sum_{i\neq n} c_n(X_i) \leq c_n(M) - \frac{n-1}{n} \cdot c_n(M) = \frac{1}{n} \cdot c_n(M) = \PROP_n.
    \end{equation*}
    
    For any other agent $i \neq n$, the design of the algorithm ensures that she receives at most her proportional share.
    Thus the algorithm computes a (complete) fractional allocation that is PROP to all agents.
\end{proof}


\begin{example} \label{example:fractional-allocation}
    Consider the following instance $\mathcal{I}^*$ with $n = 4$ agents, $m = 6$ items, and costs shown in Table~\ref{tab:exp-instance}, where $\epsilon>0$ is arbitrarily small.
    Note that the instance is IDO and we have $\PROP_1 = 1.5-\epsilon, \PROP_2 = 1.25 - \epsilon$ and $\PROP_3 = \PROP_4 = 1$.
    
    \begin{table}[htbp]
        \centering
        \begin{tabular}{c|c|c|c|c|c|c}
            &  $e_1$ & $e_2$ & $e_3$ & $e_4$ & $e_5$ & $e_6$\\ \hline
            agent 1   & $1$ & $1$ & $1$ & $1$ & $1$ & $1-4\epsilon$\\
            agent 2   & $1$ & $1$ & $1$ & $1$ & $1-4\epsilon$ & $0$\\
            agent 3   & $1$ & $1$ & $1$ & $1$ & $0$ & $0$\\
            agent 4   & $1$ & $1$ & $1$ & $1$ & $0$ & $0$ 
        \end{tabular}
        \smallskip
        \caption{Example instance $\mathcal{I}^*$ with $4$ agents and $6$ items.}
        \label{tab:exp-instance}
    \end{table}

    After running Algorithm~\ref{alg:MKA} on instance $\mathcal{I}^*$, we obtain the following fractional allocation (see Table~\ref{tab:fractional-allocation} and Figure~\ref{fig:allocation_of_I*}), where the number in each cell corresponds to the fraction of item the agent receives.
    
    \begin{table}[htbp]
        \centering
        \begin{tabular}{c|c|c|c|c|c|c}
            &  $e_1$ & $e_2$ & $e_3$ & $e_4$ & $e_5$ & $e_6$\\ \hline
            agent 1   & $1$ & $0.5-\epsilon$ & $0$ & $0$ & $0$ & $0$\\
            agent 2   & $0$ & $0.5+\epsilon$ & $0.75-2\epsilon$ & $0$ & $0$ & $0$\\
            agent 3   & $0$ & $0$ & $0.25+2\epsilon$ & $0.75-2\epsilon$ & $0$ & $0$\\
            agent 4   & $0$ & $0$ & $0$ & $0.25+2\epsilon$ & $1$ & $1$
        \end{tabular}
        \smallskip
        \caption{The fractional allocation returned by the algorithm.}
        \label{tab:fractional-allocation}
    \end{table}

    \begin{figure}[h]
        \centering
        \begin{tikzpicture}	
            \filldraw [fill = gray!50, draw = none] (0,0) rectangle (3,1);
            \filldraw [fill = gray!50, draw = none] (5.5,0) rectangle (7.5,1);
            \node at (1.5,1.5) {Agent $1$};    \node at (4.25,1.5) {Agent $2$};
            \node at (6.5,1.5) {Agent $3$};       \node at (10.25,1.5) {Agent $4$};
            \draw (0,0) rectangle (2,1); \node at (1,-0.5) {$e_1$};
            \draw (2,0) rectangle (4,1); \node at (3,-0.5) {$e_2$};
            \draw (4,0) rectangle (6,1); \node at (5,-0.5) {$e_3$};
            \draw (6,0) rectangle (8,1); \node at (7,-0.5) {$e_4$};
            \draw (8,0) rectangle (10,1); \node at (9,-0.5) {$e_5$};
            \draw (10,0) rectangle (12,1); \node at (11,-0.5) {$e_6$};      
                \footnotesize \node at (1.5,0.5) {$1$};
            \draw [densely dashed] (3,-0.25)--(3,1.25); 
                \footnotesize \node at (2.5,0.5) {$0.5 \! - \! \epsilon$};
                \footnotesize \node at (3.5,0.5) {$0.5 \! + \! \epsilon$};
            \draw [densely dashed] (5.5,-0.25)--(5.5,1.25); 
                \footnotesize \node at (4.75,0.5) {$0.75 \! - \! 2\epsilon$};
            \draw [densely dashed] (7.5,-0.25)--(7.5,1.25); 
                \footnotesize \node at (6.75,0.5) {$0.75 \! - \! 2\epsilon$};
                \footnotesize \node at (9,0.5) {$1$};
                \footnotesize \node at (11,0.5) {$1$};
        \end{tikzpicture}
        \caption{Illustration the fractional allocation returned for instance $\mathcal{I}^*$.}
        \label{fig:allocation_of_I*}
    \end{figure}
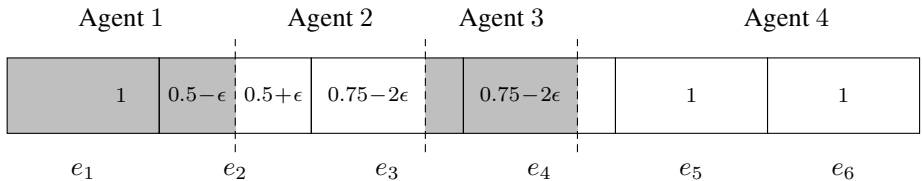
\end{example}

Since each agent receives a continuous interval in the Moving Knife Algorithm, there are at most $n-1$ cutting points.
Hence in the fractional allocation $\bX^0$, there are at most $n-1$ items that are fractionally allocated.
We call these items \emph{fractional} items.
Note that the number of fractional items can be strictly less than $n-1$, e.g., some item may get cut into three or more pieces.

\subsection{Rounding Scheme and Subsidy} \label{ssec:roudning}

In this section, we study different rounding schemes to turn the fractional allocation integral, and upper bound the total subsidy required to achieve proportionality.
We call a fractional item \emph{rounded} to some agent $i$ if, in the integral allocation, the item is fully allocated to agent $i$.
In the following, we consider the case when there are exactly $n-1$ fractional items in $\bX^0$ and we leave the case with less than $n-1$ fractional items to Section~\ref{ssec:less-than-n-1}.
Note that the analysis of these two cases are very similar and the upper bounds we derive follow the same formulation.

Given a fractional allocation $\bX^0$ returned by the Moving Knife Algorithm, we use $e_1, \ldots, e_{n-1}$ to denote the $n-1$ fractional items that are ordered by the time they are allocated.
In other words, for all $i\in \{1,\ldots, n-1\}$, item $e_i$ is shared by agents $i$ and $i+1$.
We denote by $x_i$ the fraction of item $e_i$ agent $i$ holds; consequently $1-x_i$ is the fraction agent $i+1$ holds.
Our goal is to round the fractional allocation $\bX^0$ to an integral allocation $\bX$ in which each fractional item $e_i$ is rounded to either agent $i$ or $i+1$.
The rounding result can be represented by a vector $\hat{x} = (\hat{x}_1, \ldots, \hat{x}_{n-1})$, where $\hat{x}_i \in \{0,1\}$ is the indicator of whether item $e_i$ is rounded to agent $i$.
Under the rounding result $\hat{x}$, for each agent $i\in N$, the subsidy $s_i$ required to guarantee proportionality is given by (for convenience we let $x_0 = \hat{x}_0 = x_n = \hat{x}_n = 1$)
\begin{equation*}
    s_i = \max \{(x_{i-1}- \hat{x}_{i-1})\cdot c_i(e_{i-1}) + (\hat{x}_i - x_i) \cdot c_i(e_i),0\}.
\end{equation*}

In the following, we consider two rounding schemes: \emph{up rounding} and \emph{threshold rounding}.

\begin{itemize}
    \item \textbf{Up Rounding:}
    For each $i\in \{1,\ldots,n-1\}$, we set $\hat{x}_i = 1$, i.e., round each item $e_i$ to agent $i$.
    \item \textbf{Threshold Rounding:} 
    For each $i\in \{1,\ldots,n-1\}$, we set $\hat{x}_i = 1$ if $x_i \geq 0.5$ and $\hat{x}_i = 0$ otherwise.
    In other words, we greedily round each $e_i$ to the agent who holds a larger fraction of $e_i$.
\end{itemize}

\begin{example} \label{example:upper-bounds}
    For the fractional allocation in Table~\ref{tab:fractional-allocation}, we have $x_1 = 0.5 - \epsilon$, $x_2 = x_3 = 0.75 - 2\epsilon$.
    \begin{itemize}
        \item Under the up rounding scheme, we obtain an allocation with $X_1 = \{e_1,e_2\}, X_2 = \{e_3\}, X_3 = \{e_4\}$ and $X_4 = \{e_5\}$, which implies $s_1 = 0.5+\epsilon, s_2 = s_3 = s_4 = 0$.
        \item Under the threshold rounding, we obtain an allocation with $X_1 = \{e_1\}, X_2 = \{e_2,e_3\}, X_3 = \{e_4\}$ and $X_4 = \{e_5, e_5\}$, which implies $s_1 = s_3 = s_4 = 0$ and $s_2 = 0.75 + \epsilon$.
    \end{itemize}
\end{example}

In the following, we show that at least one of the above two rounding schemes computes a PROPS allocation with total subsidy at most $n/4$, and prove the following.

\begin{theorem}\label{theorem:n/4subsidy}
    Given the fractional allocation $\bX^0$ with $n-1$ fractional items returned by Algorithm~\ref{alg:MKA}, there exists a rounding scheme that returns an integral PROPS allocation $\bX$ with total subsidy at most $n/4$.
\end{theorem}

\subsubsection{Upper Bounding the Total Subsidy Required by Up Rounding}

In the following, we derive a formula that upper bounds the total subsidy required by the up rounding in terms of $\{x_1,\ldots,x_{n-1}\}$.
We first show the following.

\begin{lemma}\label{lemma:si_upperbound}
    Under the up rounding, the subsidy $\bs = (s_1, \ldots, s_n)$ satisfies the following properties:
    \begin{itemize}
        \item $s_1 \leq 1-x_1$, $s_n = 0$;
        \item $s_i \leq \max \{x_{i-1}-x_i, 0\}$ for all $i\in \{2,\ldots, n-1\}$.
    \end{itemize}
\end{lemma}
\begin{proof}
From the rounding scheme, we directly have $s_1 = (1-x_1)\cdot c_1(e_1) \leq 1-x_1$ and $s_n = 0$.
Now fix any $i\in \{2,\ldots, n-1\}$.
Since item $e_{i-1}$ is rounded to agent $i-1$ and item $e_i$ is rounded to agent $i$ in the integral allocation $\bX$, the subsidy $s_i$ for agent $i$ is given by
\begin{align*}
    s_i &= \max \{(x_{i-1}-1) \cdot c_i(e_{i-1}) + (1-x_i) \cdot c_i(e_i),0\} \\
    &\leq \max \{(x_{i-1}-x_i) \cdot c_i(e_{i-1}), 0\} \leq \max \{x_{i-1}-x_i, 0\},
\end{align*}
where the first inequality follows from $c_i(e_{i-1}) \geq c_i(e_i)$ (since the instance is IDO) and the second inequality follows from $c_i(e_{i-1}) \leq 1$.
\end{proof}

\begin{lemma} \label{lemma:upper-bound-up-rounding}
    There exists a sequence of indices $1\leq j_1 < i_2 < j_2 < \cdots < i_z < j_z \leq n-1$ such that 
    \begin{equation*}
    \sum_{i=1}^n s_i \leq (1-x_{j_1}) + (x_{i_2} - x_{j_2}) + \cdots + (x_{i_z} - x_{j_z}).
    \end{equation*}
\end{lemma}
\begin{proof}
    By Lemma~\ref{lemma:si_upperbound}, we can upper bound the total subsidy by
    \begin{equation*}
        \sum_{i=1}^n s_i \leq 1-x_1 + \sum_{i=2}^{n-1} \max\{x_{i-1}-x_i, 0\}.
    \end{equation*}

    For convenience, we introduce $x_0 = 1$.
    For each $i\in \{1,2, \ldots, n-1\}$, we use $a_i\in \{0,1\}$ to indicate whether $x_{i-1} > x_i$.
    In other words, we have $\max \{x_{i-1}-x_i , 0\} = x_{i-1} -x_i$ if $a_i = 1$ and $\max \{x_{i-1}-x_i , 0\} = 0$ otherwise.
    Note that we always have $a_1 = 1$.
    Observe that if we have $a_l = 1$ for all $l\in \{i,i+1,\ldots,j\}$, then we have
    \begin{equation*}
        \sum_{l=i}^{j} s_l \leq \sum_{l=i}^{j} (x_{l-1}-x_l) = x_{i-1}-x_j.
    \end{equation*}
    
    Therefore, given any $(a_1,\ldots,a_{n-1})\in \{0,1\}^{n-1}$, we can break the sequence into several segments of consecutive $1$'s by removing the $0$'s.
    Specifically, let $j_1$ be the fist index such that $a_{j_1+1} = 0$; let $i_2$ be the first index after $j_1$ such that $a_{i_2+1} = 1$; let $j_2$ be the fist index after $i_2$ such that $a_{j_2+1} = 0$, etc (see Figure~\ref{fig:maximal-array} for an example).
    In other words, the sub-sequences $\{a_1,\ldots,a_{j_1}\}$, $\{a_{i_2+1},\ldots,a_{j_2}\},\ldots, \{a_{i_z+1},\ldots,a_{j_z}\}$ are the maximal segments of consecutive $1$'s of the sequence $(a_1,\ldots,a_{n-1})\in \{0,1\}^{n-1}$.
    
    \begin{figure}[htbp]
    \centering
    \begin{tikzpicture}
        \node at (-1,1.5) {Agents};
        \node at (-1,0.5) {$a_i$};
        \filldraw [fill = gray!50, draw = none] (0,0) rectangle (2,1);
        \filldraw [fill = gray!50, draw = none] (3,0) rectangle (4,1);
        \filldraw [fill = gray!50, draw = none] (6,0) rectangle (8,1);
        \draw (0,0) rectangle (1,1);
            \node at (0.5,0.5) {$1$};
            \node at (0.5,1.5) {$1$};
        \draw (1,0) rectangle (2,1);
            \node at (1.5,0.5) {$1$};
            \node at (1.5,1.5) {$2$};
            \node at (1.5,-0.5) {$j_1$};
        \draw (2,0) rectangle (3,1);
            \node at (2.5,0.5) {$0$};
            \node at (2.5,1.5) {$3$};
            \node at (2.5,-0.5) {$i_2$};
        \draw (3,0) rectangle (4,1);
            \node at (3.5,0.5) {$1$};
            \node at (3.5,1.5) {$4$};
            \node at (3.5,-0.5) {$j_2$};
        \draw (4,0) rectangle (5,1);
            \node at (4.5,0.5) {$0$};
            \node at (4.5,1.5) {$5$};
        \draw (5,0) rectangle (6,1);
            \node at (5.5,0.5) {$0$};
            \node at (5.5,1.5) {$6$};
            \node at (5.5,-0.5) {$i_3$};
        \draw (6,0) rectangle (7,1);
            \node at (6.5,0.5) {$1$};
            \node at (6.5,1.5) {$7$};
        \draw (7,0) rectangle (8,1);
            \node at (7.5,0.5) {$1$};
            \node at (7.5,1.5) {$8$};
            \node at (7.5,-0.5) {$j_3$};
        \draw (8,0) rectangle (9,1);
            \node at (8.5,0.5) {$0$};
            \node at (8.5,1.5) {$9$};
    \end{tikzpicture}
    \caption{An example for identifying the indices $j_1,i_2,j_2,\ldots,i_z,j_z$. In the example we have $n=10$ and $z = 3$. The total subsidy is upper bounded by $(1-x_2) + (x_4 - x_3) + (x_8 - x_6)$.}
    \label{fig:maximal-array}
    \end{figure}

    Therefore we can identify the indices $0 = i_1 < j_1 < i_2 < j_2 < \cdots < i_z < j_z \leq n-1$ and by a telescope sum for each segment $[i_x, j_x]$, where $x \in \{1,2,\ldots,z\}$, we obtain the claimed upper bound.
\end{proof}

\subsubsection{Upper Bounding the Total Subsidy Required by Threshold Rounding}
\label{sssec:subsidy-for-threshold-rounding}

In the following, we derive a formula that upper bounds the total subsidy required by the threshold rounding in terms of $\{x_1,\ldots,x_{n-1}\}$.
We use a charging argument that charges money to the fractional items $e_1,\ldots,e_{n-1}$: we charge each fractional item $e_i$ an amount of money $p_i = \min\{ x_i, 1-x_i \}$.
We show that the total charge to the fractional items is sufficient to pay for the subsidy.

\begin{lemma} \label{lemma:upper-bound-threshold-rounding}
    For all $i\in \{1,2,\ldots,n-1\}$, let $p_i = \min\{ x_i, 1-x_i \}$. Then under the threshold rounding we have $\|\bs\|_1 \leq \sum_{i=1}^{n-1} p_i$.
\end{lemma}
\begin{proof}
    Fix any agent $i\in N$. It suffices to show that the money we charge the fractional items that are allocated to $i$ in $\bX$ is at least $s_i$.
    Suppose item $e_{i-1}$ is rounded to agent $i$, then we have $x_{i-1} < 0.5$ and thus $p_{i-1} = x_{i-1}$.
    Note that the inclusion of (the integral) item $e_{i-1}$ to $X_i$ incurs an increase in subsidy by at most $x_{i-1}\cdot c_i(e_{i-1}) \leq x_{i-1} = p_{i-1}$.
    In other words, the money $p_{i-1}$ we charge item $e_{i-1}$ is sufficient to pay for the subsidy incurred.
    Similarly, if item $e_{i}$ is rounded to agent $i$, then the money $p_i = 1-x_i$ we charge item $e_{i}$ is sufficient to pay for the subsidy $(1-x_i)\cdot c_i(e_i)$ incurred by the item.
\end{proof}

\subsubsection{Putting the Two Bounds Together} \label{sssec:combining-two-bounds}

Finally, we put the upper bounds we derived for the two rounding schemes together to prove Theorem~\ref{theorem:n/4subsidy}.
By Lemma~\ref{lemma:upper-bound-up-rounding}, there exists $1\leq j_1 < i_2 < j_2 < \cdots < i_z < j_z \leq n-1$ such that the total subsidy required by the up rounding is at most
\begin{equation} \label{eq:up-rounding}
    (1-x_{j_1}) + (x_{i_2} - x_{j_2}) + \cdots + (x_{i_z} - x_{j_z}). 
\end{equation}

By Lemma~\ref{lemma:upper-bound-threshold-rounding}, the total subsidy required by the threshold rounding is at most $\sum_{i=1}^{n-1} p_i$, where $p_i = \min\{x_i, 1-x_i\}$.
In order to combine the two upper bounds, we apply the following relaxation on $p_i$:
\begin{itemize}
    \item For each $i\in \{i_2, \ldots, i_z\}$, we use $p_i \leq 1 - x_{i}$;
    \item For each $i\in \{j_1, \ldots, j_z\}$, we use $p_i \leq x_{i}$;
    \item For every other $i$, we use that $p_i = \min\{x_i, 1-x_i\} \leq 0.5$.
\end{itemize}
Applying the above upper bounds, we upper bound the total subsidy required by the threshold rounding by
\begin{equation}\label{eq:threshold-rounding}
    x_{j_1} + (1 - x_{i_2} + x_{j_2}) + \cdots + (1 - x_{i_z} + x_{j_z}) + (n - 2z) \cdot \frac{1}{2}.
\end{equation}

Summing the above two upper bounds~\eqref{eq:up-rounding} and~\eqref{eq:threshold-rounding}, we have that the total subsidy required by the two rounding schemes combined is at most
\begin{equation*}
    z + (n-2z)\cdot \frac{1}{2} = \frac{n}{2}.
\end{equation*}

Therefore, at least one of the two rounding schemes requires a total subsidy at most $n/4$, which proves Theorem~\ref{theorem:n/4subsidy}.

\subsection{When the Number of Fractional Items is less than \texorpdfstring{$n-1$}{}}
\label{ssec:less-than-n-1}

Next, we consider the case when the number of fractional items is strictly less than $n-1$.
This can only happen when some agents receive only one item in $\bX^0$, and the item is fractional.
By applying an analysis similar to Section~\ref{ssec:roudning}, we show that the total subsidy required is at most $n/4$.

Suppose that some item $e = e_i = \cdots = e_j$ is shared by agents $\{i, \ldots, j+1\}$, e.g., item $e$ is cut $j-i+1$ times.
We define $x_i$ to be the fraction agent $i$ holds for item $e$ and $x_l - x_{l-1}$ to be the fraction agent $l$ holds for item $e$, for all $l \in \{i+1, \ldots, j\}$ (see Figure~\ref{fig:multi-cut} for an example with $i=1$ and $j=3$).
We redefine the up rounding and threshold rounding scheme as follows.
\begin{itemize}
    \item In the up rounding, we set $\hat{x}_i = 1$, and consequently we have $\hat{x}_{i+1} = \cdots = \hat{x}_j = 0$.
    Furthermore, we have $s_{i+1} = \cdots = s_{j} = 0$.
    Therefore, it is easy to verify that Lemma~\ref{lemma:si_upperbound} and~\ref{lemma:upper-bound-up-rounding} still hold\footnote{While we define the variables $x_i$'s differently for some agents, since they require $0$ subsidy, the corresponding variables do not appear in the final upper bound on the total subsidy.}.

    \item In the threshold rounding we set $\hat{x}_t = 1$ for the agent $t\in \{i,\ldots,j+1\}$ who holds the largest fraction of item $e$.
    We charge item $e$ an amount of money $p = \min\{x_i,1-x_j\} + \frac{j-i}{2}$.\footnote{Note that when $j=i$, the item is cut exactly once and the charging is the same as we defined in Section~\ref{sssec:subsidy-for-threshold-rounding}. Thus we can interpret the charging to item $e$ as paying $0.5$ more money for every extra cut.}
    It suffices to show that $p$ is enough to pay for the subsidy item $e$ incurs under the threshold rounding to prove an analogous version of Lemma~\ref{lemma:upper-bound-threshold-rounding}.
    This is true because item $e$ is cut into $j-i+2$ fractions $\{x_i, x_{i+1}-x_i,\ldots,x_j-x_{j-1}, 1-x_{j}\}$ and agent $t$ holds the largest fraction, which implies that the total fraction not held by agent $t$ (before rounding) is at most $p = \min\{x_i,1-x_j\} + \frac{j-i}{2}$.
\end{itemize}

Therefore we can upper bound the total subsidy required by the two rounding schemes in a similar way as in Lemma~\ref{lemma:upper-bound-up-rounding} and~\ref{lemma:upper-bound-threshold-rounding}.
Following the same analysis as in Section~\ref{sssec:combining-two-bounds}, we can prove the following.

\begin{theorem}\label{theorem:<n-1items}
    Given the fractional allocation $\bX^0$ with less than $n-1$ fractional items returned by Algorithm~\ref{alg:MKA}, there exists a rounding scheme that returns an integral PROPS allocation $\bX$ with total subsidy at most $n/4$.
\end{theorem}

\begin{example} \label{example:multi-cut}
    Consider a fractional allocation in which a unit cost item $e$ is shared by agents $1,2,3,4$, i.e. $e = e_1 = e_2 = e_3$.
    We define $x_1, x_2-x_1, x_3-x_2$ be the fraction of item $e$ agents $1,2,3$ hold respectively.
    Under the up rounding, agents $2$ and $3$ do not receive any item, leading to $s_2 = s_3 = s_4 = 0$, while $s_1 = 1-x_1$.
    Under the threshold rounding, we charge item $e$ money $p = \min\{x_1, 1-x_3\}+1$.
    Obviously no matter who receives item $e$ after rounding, the money is sufficient to pay for the incurred subsidy.
    Finally, since $p \leq x_1 + 1$, the total subsidy required by the two rounding schemes combined is at most $2$, which implies that at least one of them requires a total subsidy of at most $1 = n/4$. 
\end{example}
 
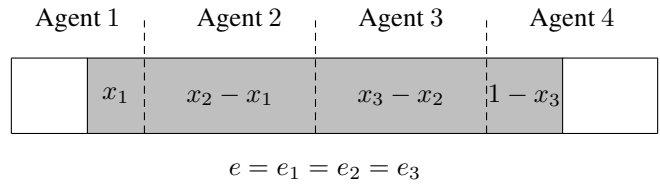
\begin{figure}[htbp]
    \centering
    \begin{tikzpicture}						
        \filldraw [fill = gray!50, draw = none] (0,0) rectangle (6.25,1);
            \node at (2,1.5) {Agent $2$};
            \node at (-0.125,1.5) {Agent $1$};  \node at (4.125,1.5) {Agent $3$};
            \node at (6.375,1.5) {Agent $4$}; 
        \draw (0,0) rectangle (6.25,1);
            \node at (3.125,-0.5) {$e = e_1 = e_2 = e_3$};
        \draw [densely dashed] (0.75,0)--(0.75,1.5); 
            \node at (0.375,0.5) {$x_1$};
        \draw [densely dashed] (3,0)--(3,1.5); 
            \node at (1.875,0.5) {$x_2-x_1$};
        \draw [densely dashed] (5.25,0)--(5.25,1.5); 
            \node at (4.125,0.5) {$x_3-x_2$};
            \node at (5.75,0.5) {$1-x_3$};
        \draw (6.25,0)--(7.5,0); \draw (6.25,1)--(7.5,1);
        \draw (0,0)--(-1,0); \draw (0,1)--(-1,1);
        \draw (-1,0)--(-1,1); \draw (7.5,0)--(7.5,1);
    \end{tikzpicture}
    \caption{Example for item $e$ being cut three times. The shadow area represents the item $e$ that is shared by four agents.}
    \label{fig:multi-cut}
\end{figure}

\subsection{Guaranteeing PROP1 without Subsidy} \label{ssec:guarantee-prop1}

Finally, we show that the PROPS allocation $\bX$ our algorithm computes is PROP1 to all agents without subsidy.
Fix any agent $i\in N$, we show that there exists $e\in X_i$ such that $c_i(X_i - e) \leq \PROP_i$.
Recall that $\bX$ is obtained by rounding the fractional PROP allocation returned by the Moving Knife Algorithm, using either up rounding or threshold rounding.

If no item or at most one item is rounded to agent $i$, the allocation PROP1 to her because by removing the item (if any) that is rounded to agent $i$, the remaining bundle has cost at most $c_i(X^0_i) \leq \PROP_i$.
Hence it remains to consider the case when both items $e_{i-1}$ and $e_i$ are rounded to agent $i$\footnote{Note that this can only happen when item $e_{i-1}$ is cut exactly once: if $e_{i-1}$ is cut at least twice then the fraction held by agent $i-1$ must be at least $1-x_{i-1}$ (otherwise agent $i$ will shout later), which implies that item $e_{i-1}$ will not be rounded to agent $i$.}.
Note that this only happens in the threshold rounding in which we greedily round each item.
Therefore, we have $x_{i-1} < 0.5$ and $x_i \geq 0.5$.
After removing the item $e_{i-1}$ we have
\begin{align*}
    c_i(X_i) - c_i(e_{i-1}) &\leq c_i(X_i) - x_{i-1} \cdot c_i(e_{i-1}) - (1-x_i) \cdot c_i(e_{i-1}) \\
    & \leq c_i(X_i) - x_{i-1} \cdot c_i(e_{i-1}) - (1-x_i) \cdot c_i(e_i)
     = c_i(X^0_i) \leq \PROP_i.
\end{align*}
where the first inequality follows from $x_{i-1} \leq 0.5$ and $1-x_i \leq 0.5$, the second inequality follows from IDO instances, and the last inequality follows from Lemma~\ref{lemma:MKA}.

\section{Agents with General Weights}\label{sec:weighted}

In this section, we consider the weighted setting in which each agent $i\in N$ has weight $w_i\ge 0$ (that represents her obligation for undertaking the chores) and $\sum_{i\in N} w_i = 1$. 
We denote by $\bw = (w_1,\dots,w_n)$ the weights of agents and $\WPROP_i = w_i\cdot c_i(M)$ the weighted proportional share of agent $i$.

\begin{definition}[WPROP and WPROPX]
    An allocation $\bX$ is called weighted proportional (WPROP) if $c_i(X_i)\le \WPROP_i$ for all $i\in N$. An allocation $\bX$ is called weighted proportional up to any item (WPROPX) if $c_i(X_i - e) \le \WPROP_i$ for any agent $i\in N$ and any item $e\in X_i$.
\end{definition}
\begin{definition}[WPROPS]
    An allocation $\bX$ with subsidies $\bs = (s_1, \ldots, s_n)$ is called weighted proportional with subsidies (WPROPS) if for any $i\in N$,
    \begin{equation*}
        c_i(X_i) - s_i \leq \WPROP_i.
    \end{equation*}
\end{definition}

In this section, we (partially) generalize our results to the weighted setting.
For identical instances, we generalize the Load Balancing Algorithm to the weighted setting and achieve the same tight bounds on the total subsidy, i.e. $n/4$ when $n$ is even and $\frac{n^2-1}{4n}$ when $n$ is odd.
For the cases with general cost functions, we show that the fractional Bid-and-Take Algorithm computes a WPROP allocation, based on which we compute a WPROPS allocation with total subsidy at most $(n-1)/2$.
Since the lower bounds in Lemma~\ref{lemma:lower-bounds} still hold for the weighted setting, in the following, we only consider the upper bounds.

\subsection{Identical Cost Functions}\label{ssec:weighted_identical}

We use the Weighted Load Balancing Algorithm to compute an allocation $\bX$ that is WPROPX (for $n$ identical agents), and show that the total subsidy required to achieve proportionality is at most $n/4$ (when $n$ is even); at most $(n^2-1)/(4n)$ (when $n$ is odd).
We can assume w.l.o.g. that $c(e_1) \geq c(e_2) \geq \cdots \geq c(e_m)$.
During the algorithm, we allocate items $e_1,e_2,\ldots,e_m$ one-by-one to the agent $i^*$ with the maximum slackness to her proportional share, i.e., $i^* = \argmax_{i\in N} \{\WPROP_i-c(X_i)\}$.
We summarize the steps of the full algorithm in Algorithm~\ref{alg:WLBA}.

\begin{algorithm}[htbp]
    \caption{Weighted Load Balancing Algorithm}
    \label{alg:WLBA}
    \KwIn{An instance $(M,N,\bw,\bc)$ with $c(e_1)\ge c(e_2)\ge \dots \ge c(e_m)$.}
    Let $X_i \gets \emptyset, \forall i\in N$\;
    \For{$j = 1,2,\dots, m$}{
        Let $i^* \gets \argmax_{i\in N} \{\WPROP_i -c(X_i)\}$\;
        Update $X_{i^*}\gets X_{i^*} + e_j$\;
        
    }
    \KwOut{An allocation $\bX = \{X_1,\dots,X_n\}$.}
\end{algorithm}

\begin{lemma}\label{lemma:weighted-identical-propx}
    The Weighted Load Balancing Algorithm (Algorithm~\ref{alg:WLBA}) computes a $\WPROPX$ allocation without subsidy for every identical instance.
\end{lemma}

\begin{proof}
    Fix any agent $i\in N$ and let $e_{\sigma(i)}$ be the last item agent $i$ receives, it suffices to show that $c(X_i - e_{\sigma(i)}) \leq \WPROP_i$ since items are allocated in the order of descending costs.
    Assume otherwise, i.e., $c(X_i - e_{\sigma(i)}) > \WPROP_i$.
    Then at the moment when agent $i$ was allocated item $e_{\sigma(i)}$, we have 
    \begin{equation*}
        \WPROP_j-c(X_j) \leq \WPROP_i - c(X_i- e_{\sigma(i)}) < 0,
    \end{equation*}
    which implies $c(X_j)>\WPROP_j$ for all $j\neq i$.
    Combining this with $c(X_i)\ge c(X_i - e_{\sigma(i)}) > \WPROP_i$, we have a contradiction that $c(M) = \sum_{i\in N} c(X_i) > \sum_{i\in N} \WPROP_i = \sum_{i\in N} \{w_i \cdot c(M)\} = c(M)$.
\end{proof}

Next, we provide upper bounds on the total subsidy required to make $\bX$ a PROPS allocation, which exactly matches the lower bounds given in Lemma~\ref{lemma:lower-bounds}.
The proof is almost the same as the proof of Theorem~\ref{theorem:identical}, with $\PROP_i$ being replaced with $\WPROP_i$.
    
\begin{theorem}\label{theorem:w n/4}
    For identical additive cost functions, there exists a WPROPS allocation with subsidy no larger than $n/4$ for even number of agents and no larger than $(n^2-1)/(4n)$ for odd number of agents.
\end{theorem}

\begin{proof}
    Given the allocation $\bX$ returned by Algorithm~\ref{alg:WLBA}, we first partition the agents into two disjoint groups:
    \begin{equation*}
        N_1 = \{i\in N : c(X_i) > \WPROP_i\}, \quad N_2 = \{i\in N : c(X_i)\le \WPROP_i\}.
    \end{equation*}

    For all $i\in N_1$, we use $e_{\sigma(i)}$ to denote the last item allocated to agent $i$.
    Then we define
    \begin{equation*}
        h_i = 
        \begin{cases}
            \WPROP_i - c(X_i - e_{\sigma(i)}), \quad &\forall i \in N_1 \\
            \WPROP_i - c(X_i), \quad & \forall i \in N_2
        \end{cases}
        .
    \end{equation*}

    By Lemma~\ref{lemma:weighted-identical-propx}, we have $h_i \geq 0$ for all $i\in N$. 
    Since we only need to subsidize agents in $N_1$, the total subsidy required to achieve proportionality can be expressed as:
    \begin{align*}
        \|\bs\|_1 & = \sum_{i\in N_1} (c(X_i) - \WPROP_i) = c(M) - \sum_{i\in N_2} c(X_i) - \sum_{i\in N_1} \WPROP_i \\
        & = \sum_{i\in N_2} \WPROP_i - \sum_{i\in N_2} c(X_i) = \sum_{i\in N_2} (\WPROP_i - c(X_i)) = \sum_{i\in N_2} h_i.
    \end{align*}

    On the other hand, we show two properties for the sequence $(h_1,\ldots,h_n)$.
    First, using the same argument as we have shown in the proof of Lemma~\ref{lemma:weighted-identical-propx}, for all $i\in N_1$ and $j\in N_2$, we have
    \begin{equation*}
        \WPROP_i - c(X_i - e_{\sigma(i)}) \geq \WPROP_j - c(X_j).
    \end{equation*}
    
    Hence we have $h_i \geq h_j$ for all $i\in N_1$ and $j\in N_2$.
    Second, by definition we have
    \begin{align*}
        \sum_{i \in N} h_i & = \sum_{i \in N_1} (\WPROP_i-c(X_i)+c(e_{\sigma(i)}))+\sum_{i \in N_2} (\WPROP_i-c(X_i)) \\
        & = \sum_{i\in N} (\WPROP_i - c(X_i)) + \sum_{i \in N_1} c(e_{\sigma(i)})
        = \sum_{i \in N_1} c(e_{\sigma(i)}) \leq |N_1|.
    \end{align*}

    Making use of the two properties, we prove the theorem: the total subsidy required is
    \begin{align*}
        \sum_{i\in N2} h_i \leq \frac{|N_2|}{n}\cdot \sum_{i\in N} h_i \leq \frac{|N_1|\cdot |N_2|}{n}\le 
        \begin{cases}
            n/4, \quad &\text{when $n$ is even}\\
            (n^2-1)/(4n),  \quad &\text{when $n$ is odd}
        \end{cases},
    \end{align*}
    where in the first inequality we use the property that $h_i \geq h_j$ for all $i\in N_1$ and $j\in N_2$.    
\end{proof}

\subsection{General Additive Cost Functions}\label{ssec:weighted_general}

In this section, we consider the case when agents have general additive cost functions.
We first show that the Moving Knife Algorithm we used in the unweighted setting (or any other similar algorithms that require each agent receiving a continuous bundle of items) fails in computing a fractional WPROP allocation in the weighted setting.

 \begin{example} \label{example:hard-WPROP}
    Consider the instance with $w_1 = 0.43$, $w_2 = 0.57$, and costs shown in Table~\ref{tab:hard-WPROP}.
    Note that for both agents $1$ and $2$ we have $c_1(M) = c_2(M) = 10$, leading to $\WPROP_1 = 4.3$ and $\WPROP_2 = 5.7$.
    For any algorithm (e.g., the Moving Knife Algorithm) that requires each agent receiving a continuous interval, we need to find a cutting position $l$ to break the interval $(0,m]$ into $(0,l]$ and $(l,m]$, and assign them to the two agents. 
    However, if agent $1$ receives $(0,l]$ then she requires $l\leq 1.3$, in which case $c_2(l,m) \geq 6.1 > 5.7 = \WPROP_2$; if agent $2$ receives $(0,l]$ then she requires $l\leq 1.9$, in which case $c_1(l,m) \geq 5.1 > 4.3 = \WPROP_1$.
    In other words, there is no cutting position to produce a WPROP allocation.
 \begin{table}[htbp]
     \centering
     \begin{tabular}{c|c|c|c|c|c|c|c}
    & $e_1$ & $e_2$ & $e_3$ & $e_4$ & $e_5$ & $e_6$ & $e_7$ \\
    \hline
    Agent $1$ & $4$ & $1$ & $1$ & $1$ & $1$ & $1$ & $1$ \\
    Agent $2$ & $3$ & $3$ & $1$ & $1$ & $1$ & $1$ & $0$ \\
     \end{tabular}
     \smallskip
     \caption{Instance showing that the Moving Knife Algorithm fails for computing a WPROP allocation of chores.}
     \label{tab:hard-WPROP}
 \end{table}
 \end{example}

Therefore, in this section, we use a fractional version of the Bid-and-Take algorithm~\cite{conf/www/0037L022} to compute a WPROP allocation, based on which we compute an integral allocation with total subsidy at most $(n-1)/2$.

\subsubsection{Fractional Bid-and-Take Algorithm}

We first introduce the notation for representing a fractional allocation.
We use $x_{ie} \in [0,1]$ to specify the fraction of item $e$ agent $i$ receives in the fractional allocation $\bX = (X_1,\ldots, X_n)$, where $X_i = (x_{ie})_{e\in M} \in [0,1]^{m}$ for each $i\in N$.
Note that the allocation is complete if and only if $\sum_{i\in N} x_{ie} = 1 $ for every $e \in M$.
In the fractional allocation $\bX$, we have $c_i(X_i) = \sum_{e \in M}x_{ie}\cdot c_i(e)$ for each agent $i\in N$.
We call $\bX$ an integral allocation if $x_{ie}\in \{0,1\}$ for every $i\in N$ and $e\in M$.

\paragraph{The Algorithm.}
We allocate the items one-by-one in a continuous manner following an arbitrarily fixed ordering $e_1,e_2,\ldots,e_m$ of the items.
Initially all agents are active.
For each item $e_j\in M$, we continuously allocate $e_j$ to the active agent $i$ with the minimum $\frac{c_i(e_j)}{c_i(M)}$, until either $e_j$ is fully allocated or $c_i(X_i) = \WPROP_i$. 
If $c_i(X_i) = \WPROP_i$, we inactive agent $i$. 
The algorithm terminates when all items are fully allocated.
The steps of the full algorithm are summarized in Algorithm~\ref{alg:FBTA}.

\begin{algorithm}[htbp]
    \caption{Fractional Bid and Take Algorithm}
    \label{alg:FBTA}
    \KwIn{An instance $(M,N,\bw,\bc)$}
    $X_i \gets \mathbf{0}^m, \forall i\in N$  \qquad \qquad \tcp{current fractional bundle}
    $A\gets N$ \qquad \qquad  \qquad \qquad \tcp{ the set of active agents}
    $\mathbf{z} \gets \mathbf{1}^m$ \qquad \qquad  \qquad \qquad\tcp{remaining fraction of the items}
    $j\gets 1$ \qquad \qquad  \qquad \qquad \> \> \tcp{item to be allocated}
    \While{$j \le m$}{

        Let $i \gets \argmin_{i'\in A} \frac{c_{i'}(e_j)}{c_{i'}(M)}$\;
        \If{$c_i(X_i) +z_j\cdot c_i(e_j)>\WPROP_i$}{
            $x_{ie_{j}} \gets \frac{\WPROP_i-c_i(X_i)}{c_i(e_j)}$\;
            $z_j \gets z_j - x_{ie_{j}}$, $A \gets A\setminus \{i\}$\;
        }
        \Else{
            $x_{ie_{j}} \gets z_{j}$, $z_{j} \gets 0$, $j\gets j+1$\;
        }    
        
    }

    \KwOut{A fractional allocation $\bX^0  = (X_1,\dots, X_n)$.}
\end{algorithm}


\begin{lemma}\label{lemma: wprop}
    The output allocation $\bX^0$ is a fractional WPROP allocation. 
\end{lemma}
\begin{proof}
    Since throughout the allocation process, we maintain the property that $c_i(X_i)\le \WPROP_i$ for each agent $i$, it suffices to show all items are fully allocated i.e., there is at least one active agent when we try to allocate each item $e$.
    Since we allocate each item $e$ to the active agent $i$ with the minimum $\frac{c_i(e)}{c_i(M)}$, we have $\frac{c_i(e)}{c_i(M)} \leq \frac{c_j(e)}{c_j(M)}$ for every active $j\neq i$. Therefore we have $\frac{c_j(X_i)}{c_j(M)} \geq \frac{c_i(X_i)}{c_i(M)}$ for any active agents $i,j\in A$.

    Assume by contradiction that when assigning some item $e$, all agents are inactive, i.e., $A=\emptyset$. Let $j$ be the last agent that becomes inactive.
    Consider the moment when $j$ becomes inactive, we have
    \begin{equation*}
        1 = \frac{c_j(M)}{c_j(M)}>\sum_{i\in N} \frac{c_j(X_i)}{c_j(M)}
        \ge \sum_{i\in N}\frac{c_i(X_i)}{c_i(M)} = \sum_{i\in N} w_i = 1,
    \end{equation*}
    which is a contradiction.
\end{proof}

Note that in the above WPROP allocation, the number of fractional items is at most $n-1$, since the number of fractional items can increase only when some agent becomes inactive, and there is at least one active agent when the algorithm terminates.
However, in the allocation an agent might receive many fractional items, which introduces difficulties in designing the rounding scheme.
Moreover, because of this, after the rounding, we can no longer guarantee that the allocation without subsidy is PROP1.

\subsubsection{Rounding Scheme and Subsidy}\label{sssec:rounding-scheme-weighted-general}
Given $\bX^0$ returned by Algorithm~\ref{alg:FBTA}, for every item $e\in M$, we use $k(e) = \{i\in N: x_{ie} > 0\}$ to denote the set of agents who gets (a fraction of) item $e$.
Note that if $|k(e)| = 1$ then item $e$ is integrally allocated to a single agent.
We call an item $e$ \emph{fractional} if and only if $|k(e)| \geq 2$.
%
We use the following rounding scheme to get an integral allocation $\bX$, based on which we analyze the upper bound of the total subsidy.

\paragraph{Rounding Scheme.}
For any fractional item $e$, let $i^* = \argmax_{i\in k(e)} \{x_{ie}\}$ be the agent who owns the maximum fraction of item $e$.
We round each item $e$ to agent $i^*$, i.e., set $x_{i^*e} = 1$ and $x_{ie} = 0$ for every $i \in k(e)\setminus \{i^*\}$.
We summarize the rounding scheme in Algorithm~\ref{alg:RS}.

\begin{algorithm}[htbp]
    \caption{Rounding Scheme}
    \label{alg:RS}
    \KwIn{A fractional allocation $\bX^0$.}
    \For{$j = 1,2,\dots, m$}{
        Let $k(e_j) = \{i\in N: x_{ie_j} > 0 \}$ \;
        \If{$|k(e_j)| \geq 2$}{
            Let $i^* \gets \argmax_{i\in k(e_j)} \{x_{ie_j}\}$, and update $x_{i^*e_j} \gets 1$\;
            Update $x_{ie_j} \gets 0$ for every $i \in k(e_j)\setminus \{i^*\}$\;
        }
    }
    Set $s_i \gets \max\{c_i(X_i) - \WPROP_i, 0\}$ for all $i\in N$\;
    \KwOut{An integral allocation $\bX  = (X_1,\dots, X_n)$ with subsidy $\bs$.}
\end{algorithm}

\begin{theorem}\label{theorem:n/2}
     There exists an algorithm that computes WPROP allocations with subsides at most $(n-1)/2$ for the allocation of chores to a group of agents having general additive cost functions.
\end{theorem}
\begin{proof}
    We denote by $e_{\sigma(i)}$ the last item allocated to agent $i$, where the allocation might be fractional, and assume w.l.o.g. that agent $n$ is the last active agent.
    Note that the set of fractional items is a subset of $F = \{ e_{\sigma(1)}, e_{\sigma(2)},\ldots, e_{\sigma(n-1)} \}$, and $F$ might be a multiset.
    Similar to the analysis in Section~\ref{sssec:subsidy-for-threshold-rounding}, we use a charging argument that charges money to the fractional items.
    We first show that the charged money is sufficient to pay for the subsidy, and then provide an upper bound on the total money we charged.

    We charge an amount of money $p(e) = \frac{|k(e)| - 1}{|k(e)|}$ to each fractional item $e$.
    Since the inclusion of item $e$ to $X_i$ incurs an increase in subsidy by at most $(1-x_{ie}) \cdot c_i(e) \leq 1-x_{ie} \leq 1 - \frac{1}{|k(e)|}$ (recall that $i$ holds the maximum fraction of item $e$), clearly the charged money is sufficient to pay for the subsidy.
    Next we upper bound the total money we charge to the fractional items.

    Consider any fractional item $e$ with $|k(e)| = k \geq 2$.
    Since $e$ is shared by $k$ agents, we know that $e$ is cut $k-1$ times, and thus there exists at least $k-1$ agents $i$ with $e_{\sigma(i)} = e$.
    We re-charge the money $p(e) = \frac{k - 1}{k}$ to these agents, where each agent is charged an amount at most $\frac{1}{k} \leq \frac{1}{2}$.
    Note that over all fractional items, each agent is charged at most once (by item $e_{\sigma(i)}$) and agent $n$ is not charged, which implies that the total money is at most $(n-1)/2$.
\end{proof}

\section{Conclusion and Open Questions} \label{sec:conclusion}

In this paper, we provide a precise characterization for the total subsidy required to achieve proportionality (for both the allocation of goods and chores).
We show that a total subsidy of $n/4$ suffices to ensure the existence of a PROP allocation, and this is the (nearly) optimal guarantee.
As we will show in the appendix, the above results extend to the allocation of goods, and partially to the weighted case.

Our work leaves many interesting questions open.
The first obvious open question is to complete our work by filling in the gaps in our results.
For example, for an odd number of agents having general additive cost functions, our lower bound $(n^2-1)/(4n)$ and upper bound $n/4$ do not match each other.
This comes from the choice of rounding schemes: it can be shown that there exists an instance for which both up rounding and threshold rounding require a total subsidy of $n/4$, even when $n$ is odd.
Whether we can improve the upper bound by introducing more rounding schemes is an interesting open problem.
Another natural open question is to improve the upper bound $(n-1)/2$ we prove for the weighted case (in Section~\ref{ssec:weighted_general}), e.g., to $n/4$.
We remark that generalizing our analysis for the unweighted setting (in Section~\ref{sec:general}) to the case when agents have arbitrary weights might be highly non-trivial.
One major difficulty comes from the computation of a well-structured fractional PROP allocation.
In fact, if we require that the bundle for each agent must form a continuous interval (as Moving Knife Algorithm computes), then such allocation might not exist (see Section~\ref{ssec:weighted_general} for an example).
It remains unknown whether other fractional PROP allocation, e.g., the one returned by the Fractional Bid-and-Take we showed in Section~\ref{ssec:weighted_general}, admits rounding schemes that require a total subsidy at most $n/4$.
Finally, we believe that it would be interesting to further extend our results to the setting of mixed items (goods and chores), or to study other fairness criteria with subsidy, e.g., Maximin Share (MMS).
It would be very interesting to investigate whether a total subsidy strictly less than $n/4$ is sufficient to guarantee the existence of MMS allocations.

\newpage
\appendix

\section{Envy-Free Allocation with at most \texorpdfstring{$n-1$}{} Total Subsidy for Chores}
\label{sec:ef-chores}

For the allocation of goods, Brustle et al.~\cite{conf/sigecom/BrustleDNSV20} show that there exists an envy-free allocation with total subsidies no larger than $n-1$ while each agent receives a subsidy at most $1$ dollar.
We show that a similar result also holds for the allocation of chores.

\begin{definition}[EF and EF1]
    An allocation $\bX$ is called envy-free (EF) if $c_i(X_i) \leq c_i(X_j)$ for any $i,j\in N$.
    An allocation $\bX$ is envy-free up to one item (EF1) if for any $i,j \in N$, there exists $e \in X_i$ such that $c_i(X_i - e)\leq c_i(X_j)$.
\end{definition}

As before, we use $s_i$ to denote the subsidy we offer to agent $i\in N$. 
The notations $\bs = (s_1, \ldots, s_n)$ and $\|\bs\|_1 = \sum_{i\in N} s_i$ are defined in the same way as in previous sections.

\begin{definition}[EFS]
    An allocation $\bX$ with subsides $\bs$ is called envy-free with subsidies (EFS) if $c_i(X_i) - s_i \leq c_i(X_j) - s_j$ for any $i,j\in N$.
\end{definition}

Our main result in this section is summarized as follows.

\begin{theorem} \label{theorem:EFS}
    For the allocation of chores, there exists an EFS allocation with total subsidy at most $n-1$, where the allocation without subsidy is EF1, the subsidy to each agent is at most $1$ and the total subsidy is at most $n-1$.
\end{theorem}

An allocation $\bX$ is called envy-freeable if there exists subsidies $\bs$ with which $\bX$ is EFS.

\begin{definition}[Envy Graph]
    Given any allocation $\bX$, we define the corresponding envy graph $G_X$ as a complete directed graph with vertex set $N$ where each arc $(i,j)$ has weight $w_X(i,j) = c_i(X_i) - c_i(X_j)$, representing the envy from agent $i$ to agent $j$ under the allocation $\bX$. For any path $P$ or cycle $C$ in the envy graph, the weight of the $P$ or $C$ is the sum of weights of arcs along $P$ or $C$.
    And we use $l(i)$ to denote the maximum weight of any path starting from vertex $i$ in envy graph $G_X$.
\end{definition}

Halpen and Shah \cite{conf/sagt/HalpernS19} characterize envy-freeable allocations by using the envy graph defined above. The following lemma summarizes Theorems 1 and 2 of \cite{conf/sagt/HalpernS19}.
while their result is stated for the allocation of goods, it is straightforward to check that the same result holds for the allocation of chores.

\begin{lemma} \label{lemma:HS19}
    Given any allocation $\bX = (X_1,\dots,X_n)$ and $q \geq 0$, the following statements are equivalent:
    \begin{itemize}
        \item [(a)] $\bX$ is envy-freeable with a subsidy of at most $q$ for each agent.
        \item [(b)] $G_X$ has neither a positive-weight cycle nor a path with a weight larger than $q$. 
    \end{itemize}
    Moreover, when the above statements hold, setting $s_i = l(i)$ for every agent $i$ yields an EFS allocation.
\end{lemma}
    

Given any instance of indivisible chores, we aim to find EFS allocation $\bX$ with limited total subsidies $\|\bs\|_1$. 
We first show that the lower bound on the total subsidy is at least $n-1$, i.e., $\|\bs\|_1\ge n-1$.

\begin{lemma}
    There exists an instance with $n$ agents for which every EFS allocation requires a total subsidy of at least $n-1$.
\end{lemma}
\begin{proof}
    Consider an instance with $n-1$ items having cost $1$ to $n$ identical agents. 
    For any allocation $\bX$, every agent requires a subsidy of $s_i \geq c_i(X_i)$ since there is at least one agent who receives nothing.
    Hence we have $\|\bs\|_1 \ge \sum_{i \in N} c_i(X_i) = n-1$. 
\end{proof}


In the following, we consider the minimum subsidy to achieve EFS allocation for indivisible chores. 
We follow a proof framework proposed by Brustle et al.~\cite{conf/sigecom/BrustleDNSV20} to show that one dollar for each agent is enough for achieving envy-freeness for the allocation of chores.
The total subsidy can be bounded by $n-1$ since there is at least one agent who gets a subsidy of $0$ (otherwise we can decrease the subsidy of all agents by the same amount).

\subsection{Algorithm for Computing an EFS Allocation}
We first show that the Modified Bounded-Subsidy Algorithm computes an EF1 and envy-freeable allocation, based on which we derive the upper bound of the total subsidy.

\paragraph{The Algorithm.}
We first define a cost graph $H$ as a complete bipartite graph on vertex sets $N$ and $M$, where each edge $(i,e_k)$ between $i\in N$ and $e_k\in M$ has weight $c_i(e_k)$.
The algorithm proceeds in rounds, where in each round we do a minimum-weight perfect matching between all agents and the unallocated items, e.g., agent $i$ matches (receives) an item $e_i^t$ in round $t$.
By introducing $\lceil m/n \rceil\cdot n - m$ dummy items with cost $0$ to all agents, we can assume w.l.o.g. that the algorithm has $T = m/n$ rounds.
Note that all dummy items (if any) will be matched in the first round.
The steps are summarized in Algorithm~\ref{alg:MBA}. 

\paragraph{Notations.}
We use $e^t_i$ to denote the item agent $i$ receives in round $t$.
Thus $X_i = \{ e^1_i, e^2_i, \ldots, e^T_i \}$.
We use $\mathcal{M}^t = \{(i,e_i^{t})\}_{i\in I}$ to denote the matching we compute in round $t$.
We use $K_t$ to denote the set of unallocated items at the beginning of round $t$ and $H[I,K]$ to denote the subgraph of $H$ where $I\subseteq N$ and $K \subseteq M$.
Specifically, we have $K_1 = M$ and $K_{t+1} = K_t \setminus \cup_{i \in N} \{e_i^{t}\}$ for all $t \in \{1, \ldots, T-1\}$.

\begin{algorithm}[htbp]
    \caption{Modified Bounded-Subsidy Algorithm} 
    \label{alg:MBA}
    \KwIn{An instance $(M,N,\bc)$}
        Initialize $X_i \gets \emptyset$ for all $i\in N$\;
        Set $t \gets 1$ and $K_1 \gets M$\;
        \While{$K_t \ne \emptyset$}{
            Compute a minimum-weight perfect matching  in $H[I,K_t]$\;
            Let $X_i \gets X_i + e_i^{t}$ for all $i \in N$\;
            Let $K_{t+1} \gets K_t \setminus \cup_{i \in N} \{e_i^{t}\}$\;
            Update $t\gets t+1$\;
        }
    \KwOut{Allocation $\bX = (X_1,\dots, X_n)$}
\end{algorithm}

    
\begin{lemma}
    Algorithm~\ref{alg:MBA} computes an EF1 allocation without subsidy. 
\end{lemma}
\begin{proof}
    Fix any agent $i\in N$.
    Note that for all $t< T$ we have $c_i(e_i^t)\le c_i(e)$ for any item $e\in K_{t+1}$ as otherwise we can improve the minimum-weight perfect matching $\mathcal{M}^t$ by changing $e_i^t$ to $e$.
    Hence for every agent $j \in N$, we have
    \begin{align*}
        c_i(X_i-e_i^T) = c_i(\{e_i^1,\dots,e_i^{T-1}\}) 
        \leq c_i(e_j^2)+\cdots+c_i(e_j^{T}) \leq c_i(X_j),
    \end{align*}
    which shows that the output allocation $\bX$ is EF1.
\end{proof}

\begin{lemma}\label{lemma:EFable}
The allocation $\bX$ returned by Algorithm~\ref{alg:MBA} is an envy-freeable allocation.
\end{lemma}
\begin{proof}
    By Lemma~\ref{lemma:HS19}, it suffices to show that there is no directed positive weight cycle in the envy graph constructed by the final allocation $\bX$.
    Take any directed cycle $C$ (represented by a collection of arcs) in envy graph $G_X$.
    Thus we have
    \begin{align*}
            w_X(C) = \sum_{(i,j)\in C} w_X(i,j) = \sum_{(i,j)\in C} (c_i(X_i)- c_i(X_j))
            = \sum_{t=1}^T \sum_{(i,j)\in C} (c_i(e_i^t)- c_i(e_j^t)).
    \end{align*}
    
    Recall that in each round $t$ we compute a minimum weight perfect matching $\mathcal{M}^t$, which implies
    \begin{equation*}
        \sum_{(i,j)\in C} (c_i(e_i^t)- c_i(e_j^t))<0,
    \end{equation*}
    as otherwise we can improve the matching $\mathcal{M}^t$ by rotating the items in the reverse order along the cycle, e.g., each agent $i$ receives item $e_j^t$ for all $(i,j) \in C$.
    Hence we have $w_X(C)<0$ as claimed.
\end{proof}

\subsection{Upper Bounds on the Total Subsidy}

In this section, we define the set of subsidies to the agents, showing that each agent receives a subsidy at most $1$, and the resulting allocation is EF.
As in~\cite{conf/sigecom/BrustleDNSV20}, we introduce a constructed cost function $\Bar{\bc}$ as follows.
Fix any agent $i\in N$, we define a constructed cost function $\Bar{\bc}$ based on $\bc$ using the following rules:
\begin{align*}
    && \bar{c_i}(e_i^t) &= c_i(e_i^t) & \forall t \le T && \\
    && \bar{c_i}(e_j^t) &= \min \{c_i(e_j^t),c_i(e_i^{t+1})\} &\forall j \in N \setminus \{i\}, \forall t \le T-1 && \\
    && \bar{c_i}(e_j^T) &= c_i(e_j^T) & \forall j \in N \setminus \{i\} &&.
\end{align*}

The following two observations are trivial but useful.
Intuitively speaking, by defining the new cost function, we amplify the envy from each agent $i$ to other agents.
By defining the subsidy based on the new cost function, we show that all envies (under the original cost function) can be eliminated.

\begin{observation}\label{obs:=}
    For any agent $i$ and item $e\in X_i$, we have $\bar{c}_i(e) = c_i(e)$.
\end{observation}

\begin{observation}\label{obs:>=}
    For any agent $i$ and item $e\notin X_i$, we have $\bar{c}_i(e)\leq c_i(e)$.
\end{observation}

Then based on $\bar{\bc}$, we define a new envy graph $\Bar{G}_X$ where each edge $(i,j)$ has a weight $\bar{w}_X(i,j) = \Bar{c}_i(X_i) - \Bar{c}_i(X_j)$.
We denote by $\bar{l}(i)$ the maximum weight of any path starting at vertex $i$ in $\Bar{G}_X$. We first show the output allocation is also envy-freeable under $\bar{\bc}$.

\begin{lemma}\label{lemma:EFable_bar}
The allocation $\bX$ is envy-freeable under the constructed cost function vector $\bar{\bc}$.
\end{lemma}
\begin{proof}
    Observed that the weight of any arc in $\Bar{G}_X$ is at least that in $G_X$ (by the above two observations).
    If all the weights of all arcs are unchanged, then we are done due to Lemma~\ref{lemma:EFable}.
    Now consider a cycle $C$ that contains at least one arc whose weight is increased.
    Assume for the sake of contradiction that $C$ has a positive weight.
    We decompose $C$ into $d$ directed paths $\{P_1, \dots, P_d\}$  where the last arc in each path is an increased-weight arc. 
    Since $C$ has a positive weight, there exists a path $P\in \{P_1,\dots, P_n\}$, say $P = (1,2,\ldots,k+1)$, of positive weight.
    Given any item set $K \subseteq M$, we define $\bar{w}_K(i,j) = \bar{c}_i(X_i \cap K) - \bar{c}_i(X_j \cap K)$. 
    Let $M^t = \{ e^t_1,\ldots,e^t_n \}$.
    Since $\bar{w}_X(P) > 0$, there must exists $t$ such that
    \begin{equation}\label{ineq:w(P)}
        \bar{w}_{{M}^t}(P) = \sum_{i=1}^k \bar{w}_{{M}^t}(i,i+1)>0.
    \end{equation}
    
    We show that there exists another perfect matching in round $t$ whose weight is strictly smaller than $\mathcal{M}^t$, which leads to a contradiction.    
    We construct a new matching $\hat{\mathcal{M}}^t = \{(i,\omega_i^t)\}_{i\in N}$ based on $\mathcal{M}^t$ in the following ways:
    \begin{align*}
    && \omega_i^t &= e_i^t & \forall i \ge k+1 && \\
    && \omega_i^t &= e_{i+1}^{t} & \forall i \le k-1 && \\
    && \omega_k^t &= e_k^{t+1} & \forall i=k. &&
    \end{align*}
    Then we have 
    \begin{align*}
        \sum_{i=1}^k (c_i(e_i^t)-c_i(\omega_i^t)) &=\sum_{i=1}^{k-1} (c_i(e_i^t)-c_i(e_{i+1}^t))+(c_k(e_k^t) - c_k(e_k^{t+1}))\\
        & = \sum_{i=1}^{k-1} (\bar{c}_i(e_i^t)-\bar{c}_i(e_{i+1}^t))+(\bar{c}_k(e_k^t) - \bar{c}_k(e_{k+1}^{t}))\\
        & = \sum_{i=1}^k \bar{w}_{M^t}(i,i+1)=\bar{w}_{M^t}(P)>0
    \end{align*}
    where the second inequality holds since in path $P$, only the last edge has its weight increased and the last inequality holds due to Equation~\ref{ineq:w(P)}.
    The inequality $\sum_{i=1}^k (c_i(e_i^t)-c_i(\omega_i^t)) >0$ implies that $\mathcal{M}^t$ is not the minimum-weight perfect matching in round $t$, which is a contradiction.
\end{proof}

Next, we set $s_i = \bar{l}(i)$ and show that the resulting allocation $\bX$ with subsidy $\bs$ is EFS.

\begin{lemma}\label{lemma:allocation-EFS}
    The allocation $\bX$ is an EFS allocation with subsidy ${\bs}$.
\end{lemma}
\begin{proof}
    Note that we have ${s}_i = \bar{l}(i) \ge 0$ for all $i\in N$ since the empty path has weight zero.
    Fix any pair of agents $i,j$.
    We have $\bar{l}(i) \geq \bar{w}_X(i,j) + \bar{l}(j)$ because (1) if arc $(i,j)$ together with the path defining $\bar{l}(j)$ form a cycle, then the RHS is non-positive, by Lemma~\ref{lemma:EFable_bar}; (2) if arc $(i,j)$ together with the path defining $\bar{l}(j)$ form a path, then the inequality holds by definition of $\bar{l}(i)$.
    Hence we have
    \begin{equation*}
        {s}_i = \bar{l}(i)\ge \bar{w}_X(i,j) + \bar{l}(j) = \bar{c}_i(X_i)- \bar{c}_i(X_j)+ {s}_j \ge c_i(X_i)- c_i(X_j)+ {s}_j.
    \end{equation*}
    where the second inequality holds due to Observation~\ref{obs:=} and~\ref{obs:>=}. 
\end{proof}

Finally, we show that the subsidy to each agent is at most $1$, which shows that the total subsidy is at most $n-1$, and completes the proof of Theorem~\ref{theorem:EFS}.

\begin{lemma}\label{lemma:upper-bound-of-s-efs}
    For all $i\in N$ we have $s_i \leq 1$.
\end{lemma}
\begin{proof}
    Recall that $s_i = \bar{l}(i)$, the weight of maximum-weight path starting from vertex $i$. Suppose the path ends at vertex $j$.
    Note that
    \begin{align*}
        \bar{w}_X(j,i) &= \bar{c}_j(X_j)- \bar{c}_j(X_i) = \sum_{t=1}^T\bar{c}_j(e_j^t)- \sum_{t=1}^T \bar{c}_j(e_i^t) \\
        &= \sum_{t=1}^T c_j(e_j^t)- \sum_{t=1}^{T-1} \min \{c_j(e_i^t), c_j(e_j^{t+1})\} - c_j(e_i^T) \\
        &\ge \sum_{t=1}^T c_j(e_j^t)- \sum_{t=1}^{T-1} c_j(e_j^{t+1}) - c_j(e_i^T)
        = c_j(e_j^1) - c_j(e_i^T) \ge -1.
    \end{align*}
    
    By Lemma \ref{lemma:EFable_bar}, every cycle in $\bar{G}_X$ has non-positive weight.
    Thus we have $s_i = \bar{l}(i) \le 1$, as otherwise the cycle formed by the path defining $\bar{l}(i)$ and arc $(j,i)$ has positive weight.
\end{proof}

\section{PROPS Allocation with at most \texorpdfstring{$n/4$}{} Total Subsidy for Goods}
\label{sec:goods}

In this section, we consider the allocation of goods and propose an algorithm that computes a PROPS allocation with total subsidy at most $n/4$.
For the case of weighted agents, we show that the total subsidy is at most $(n-1)/2$.
Since the analysis is almost identical to the ones we have shown in previous sections, we will only highlight the main ideas and changes to the proof, without presenting too much tedious and repetitive analysis.

\subsection{The Notations and Definitions}

We consider the problem of allocating $m$ indivisible goods $M$ to $n$ agents $N$ where each agent $i\in N$ has weight $w_i > 0$ and additive valuation function $v_i:2^M \to \bR^+ \cup \{0\}$.
As before, we assume $\sum_{i\in N} w_i = 1$ and $v_i(e) \leq 1$ for all $i\in N$, $e\in M$.
We define $\WPROP_i$ as agent $i$'s proportional share, i.e., $\PROP_i = w_i\cdot v_i(M)$.

\begin{definition}[WPROP and WPROP1]
    An allocation $\bX$ is called weighted proportional (WPROP) if $v_i(X_i) \geq \WPROP_i$ for all $i\in N$.
    An allocation $\bX$ is called weighted proportional up to one item (WPROP1) if for any $i\in N$, there exists an item $e\in M\setminus X_i$ such that $v_i(X_i + e) \geq \WPROP_i$.
\end{definition}

As before, we use $s_i \geq 0$ to denote the subsidy we give to agent $i\in N$, $\bs = (s_1, \ldots, s_n)$ and $\|\bs\|_1 = \sum_{i\in N} s_i$.

\begin{definition}[WPROPS]
    An allocation $\bX$ with subsidies $\bs = (s_1, \ldots, s_n)$ is called weighted proportional with subsidies (WPROPS) if for any $i\in N$,
    \begin{equation*}
        v_i(X_i) + s_i \geq \WPROP_i.
    \end{equation*}
\end{definition}

Given any allocation $\bX$, computing the minimum subsidy to achieve weighted proportionality can be trivially done by setting
\begin{equation*}
    s_i = \max \{\WPROP_i - v_i(X_i), 0\}, \qquad \forall i\in N.
\end{equation*}

When $w_i = 1/n$ for all $i\in N$, we consider the instance unweighted, in which case the above notations and definitions become $\PROP_i$, PROP, PROP1 and PROPS.

\subsection{Lower Bounds on the Total Subsidy}\label{ssec:LB-goods}

We first provide the lower bounds on the total subsidy for guaranteeing proportionality for goods.

\begin{lemma}\label{lemma:lower-bound-goods}
    Given any $n\geq 2$, there exists an instance for the allocation of goods with $n$ agents for which every PROPS allocation requires a total subsidy of at least $n/4$ (when $n$ is even); at least $(n^2-1)/(4n)$ (when $n$ is odd).
\end{lemma}
\begin{proof}
Suppose $n\geq 2$ is even.
Consider the instance with $n$ agents and $n/2$ items where each item values $1$ to all agents.
For every agent $i\in N$, her proportional share is $\PROP = 1/2$.
Consider any allocation $\bX$ and suppose that $k \leq n/2$ agents receive at least one item.
Then each agent who does not receive any item requires a subsidy of $1/2$, which implies $\|\bs\|_1 = (n - k)/2 \geq n/4$.
In other words, any PROPS allocation requires a total subsidy of at least $n/4$.

Suppose $n\geq 2$ is odd.
Consider the instance with $n$ agents and $(n-1)/2$ items where each item values $1$ to all agents. 
For every agent $i\in N$, her proportional share is $\PROP = (n-1)/(2n)$.
Following a similar analysis as above we can show that the total subsidy required by any PROPS allocation is at least $(n+1)/2 \cdot (n-1)/(2n) \geq (n^2-1)/(4n)$.
Therefore, any PROPS allocation requires a total subsidy of at least $(n^2-1)/(4n)$.
\end{proof}

\subsection{Upper Bound on the Total Subsidy: Unweighted Case}

In this section, we consider the unweighted case, i.e., $w_i = 1/n$ for all $i\in N$.
We first show a similar reduction as Lemma~\ref{lemma:reduction-to-IDO} which allows us to consider only IDO instances.
Since the proof is almost identical as that for Lemma~\ref{lemma:reduction-to-IDO}, we omit it.

\begin{definition}[Identical Ordering (IDO) Instances]
    An instance is called identical ordering (IDO) if all agents have the same ordinal preference on the items, i.e., $v_i(e_1) \geq v_i(e_2) \geq \cdots \geq v_i(e_m)$ for all $i\in N$.
\end{definition}

\begin{lemma}
    If there exists a polynomial time algorithm that given any IDO instance computes a PROPS allocation with at most $\alpha$ subsidy, then there exists a polynomial time algorithm that given any instance computes a PROPS allocation with at most $\alpha$ subsidy.
\end{lemma}

With the above reduction, in the following, we only consider the IDO instances. Like the chores setting, our algorithm has two main steps: we first compute a fractional PROP allocation, in which a small number of items are fractionally allocated; then we find a way to round the fractional allocation to an integral one. Since some agents may have bundle value less than their proportional share after rounding, we offer subsidies to these agents. By carefully deciding the rounding scheme, we show that the total subsidy required is at most $n/4$.

\subsubsection{Computing an Allocation with at most \texorpdfstring{$n-1$}{} Factional Items}

In this section, we use the Moving Knife Algorithm to compute a fractional PROP allocation. 

\paragraph{The Algorithm.}
As before, we interpret the $m$ items as in an interval $(0,m]$, where item $e_i$ corresponds to interval $(i-1, i]$.
We interpret every interval as a bundle of items, where some items might be fractional.
Specifically, interval $(l,r]$ contains $(\lceil l \rceil - l)$-fraction of item $e_{\lceil l \rceil}$, $(r - \lfloor r \rfloor)$-fraction of item $e_{\lceil r \rceil}$ and integral item $e_j$ for every integer $j$ satisfying $(j-1,j] \subseteq (l,r]$.
The value of the interval to each agent $i\in N$ is also defined in the natural way:
\begin{equation*}
    v_i(l,r) = (\lceil l \rceil - l)\cdot v_i(e_{\lceil l \rceil}) + \sum_{j = \lceil l \rceil + 1}^{\lfloor r \rfloor} v_i(e_j) + (r - \lfloor r \rfloor)\cdot v_i(e_{\lceil r \rceil}).
\end{equation*}
The algorithm proceeds in rounds, where in each round some agent picks an interval and leaves.
We maintain that at the beginning of each round, the remaining set of items forms a continuous interval $(l,m]$.
In each round, we imagine that there is a moving knife that moves from the leftmost position $l$ to the right.
Each agent shouts if she thinks that the value of the interval passed by the knife is equal to her proportional share.
The first agent who shouts picks the interval passed by the knife and leaves, and the algorithm recurs on the remaining interval.
If in some round there is only one agent who has not left, she receives the whole remaining interval $(l,m]$ and the algorithm terminates.
The steps of the full algorithm are summarized in Algorithm~\ref{alg:MKA-goods}.

\begin{algorithm}[htbp]
    \caption{The Moving Knife Algorithm}
    \label{alg:MKA-goods}
    \KwIn{The interval $(0,m]$ corresponding to all items $M$, agents $N$, valuation functions $\mathbf{v} = (v_1, \ldots, v_n)$.}
    Initialize $X^0_i \gets \emptyset$ for each $i \in N$, and let $l \gets 0$\;
    \While{$|N| \geq 2$}{
        Let $r_i \gets \min \{r \leq m: v_i(l,r) \geq v_i(M)/n\}$ for all $i \in N$\;
        Let $i^* \gets \argmin \{r_i\}$\;
        Update $X^0_{i^*} \gets (l,r_{i^*}]$\;
        Update $N \gets N\setminus \{i^*\}$, $l \gets r_{i^*}$\;
    }
    Update $X^0_{i} \gets (l,m]$ for the unique agent $i\in N$\;
    \KwOut{Fractional allocation $\bX^0 = (X^0_1, \ldots, X^0_n)
    $.}
\end{algorithm}

By renaming the agents, we assume w.l.o.g. that agents are indexed by their picking order, i.e., agent $i$ is the $i$-th agent who picks and leaves.

\begin{lemma}\label{lemma:MKAg}
     The Moving Knife Algorithm computes fractional PROP allocations in polynomial time.
\end{lemma}

\begin{proof}
    We consider the agent $n$ who receives the last bundle of items.
    For any other agent $i \neq n$ we must have $v_n(X_{i}) \leq \frac{1}{n} \cdot v_n(M)$ since at the round that agent $i$ picks a bundle, agent $n$ shouts after agent $i$, otherwise agent $n$ would take the bundle instead.
    Hence we have
    \begin{equation*}
        v_n(X_n) = v_n(M) - \sum_{i\neq n} v_n(X_i) \geq v_n(M) - \frac{n-1}{n} \cdot v_n(M) = \frac{1}{n} \cdot v_n(M) = \PROP_n.
    \end{equation*}
    For any other agent $i \neq n$, she exactly receives her proportional share of items, i.e., $v_i(X_i) = \frac{1}{n} \cdot c_i(M) = \PROP_i$.
    In conclusion, the Moving Knife Algorithm computes a (complete) fractional allocation that is PROP to all agents.
\end{proof}

Since each agent receives a continuous interval in the Moving Knife Algorithm, there are at most $n-1$ cutting points.
Hence in the fractional allocation $\bX^0$, there are at most $n-1$ items that are fractionally allocated.
We call these items \emph{fractional} items.
Note that the number of fractional items can be strictly less than $n-1$, e.g., some item may get cut into three or more pieces.

\subsubsection{Rounding Scheme and Subsidy}\label{sssec:rounding-goods}

As for the allocation of chores, in this section, we study consider two rounding schemes (down rounding and threshold rounding).
In the following we only consider the cases when there are exactly $n-1$ fractional items as an example; the case when some items are cut more than once can be handled in the same way as we have illustrated in Section~\ref{ssec:less-than-n-1} and thus is omitted.
We use $e_1, \ldots, e_{n-1}$ to denote the $n-1$ fractional items that are ordered by the time they are allocated.
In other words, we have that item $e_i$ being shared by agent $i$ and $i+1$ for any $i\in \{1,\ldots, n-1\}$.
We denote $x_i$ as the fraction agent $i$ holds for item $e_i$ and $1-x_i$ as the fraction agent $i+1$ holds for item $e_i$ in $\bX^0$.
Based on the fractional allocation $\bX^0$, we round to an integral allocation $\bX$ in which each fractional item $e_i$ is rounded to agent $i$ or $i+1$.
A rounding scheme can be represented by a vector $\hat{x} = (\hat{x}_1, \ldots, \hat{x}_{n-1})$, in which each $\hat{x}_i \in \{0,1\}$ is the indicator of whether item $e_i$ is rounded to agent $i$.
For each agent $i$, the minimum subsidy $s_i$ to guarantee PROP for her can be formulated as
\begin{equation*}
    s_i = \max \{(\hat{x}_{i-1} - x_{i-1})\cdot v_i(e_{i-1}) + (x_i - \hat{x}_i) \cdot v_i(e_i), 0\}.
\end{equation*}

We consider two different rounding schemes: \emph{down rounding} and \emph{threshold rounding}.

\begin{itemize}
    \item \textbf{Down Rounding:}
    For all item $e_i$, we set $\hat{x}_i = 0$, e.g., rounding each item $e_i$ ``down'' to agent $i+1$.
    \item \textbf{Threshold Rounding:} 
    For each $i\in \{1,\ldots,n-1\}$, we set $\hat{x}_i = 1$ if $x_i \geq 0.5$ and $\hat{x}_i = 0$ otherwise.
    In other words, we greedily round each $e_i$ to the agent that holds more fraction of $e_i$.
\end{itemize}


\begin{lemma}
    Under the down rounding, the subsidy $\bs = (s_1, \ldots, s_n)$ satisfies the following properties:
    \begin{itemize}
        \item $s_1 \leq x_1$, $s_n = 0$;
        \item $s_i \leq \max \{x_i - x_{i-1}, 0\}$ for all $i\in \{2,\ldots, n\}$.
    \end{itemize}
\end{lemma}
\begin{proof}
From the rounding scheme, we directly have $s_1 = x_1 \cdot c_1(e_1) \leq x_1$ and $s_n = 0$.
Now fix any $i\in \{2,\ldots, n-1\}$, we have
\begin{align*}
    s_i &= \max \{(0-x_{i-1}) \cdot v_i(e_{i-1}) + (x_i-0) \cdot v_i(e_i),0\} \\
    &\leq \max \{(x_i-x_{i-1}) \cdot v_i(e_{i-1}), 0\} \leq \max \{x_{i}-x_{i-1}, 0\},
\end{align*}
where in the first inequality we use $v_i(e_{i-1}) \geq v_i(e_i)$ and in the second inequality we use $v_i(e_{i-1}) \leq 1$.
\end{proof}

Given the above lemma, following the same analysis as in the proof of Lemma~\ref{lemma:upper-bound-up-rounding}, we can derive the following bound on the total subsidy required by the down rounding (the proof is omitted).

\begin{lemma} \label{lemma:upper-bound-down-rounding}
    There exists a sequence of indices $1\leq j_1 < i_2 < j_2 < \cdots < i_z < j_z \leq n-1$ such that 
    \begin{equation*}
    \sum_{i=1}^n s_i \leq x_{j_1} + (x_{j_2} - x_{i_2}) + \cdots + (x_{j_z} - x_{i_z}).
    \end{equation*}
\end{lemma}

Next we upper bound the total subsidy required by the threshold rounding in terms of $\{x_1,\ldots,x_{n-1}\}$.
We use a charging argument that charges money to the fractional items $e_1,\ldots,e_{n-1}$.
Specifically, we charge each fractional item $e_i$ an amount of money $p_i = \min\{ x_i, 1-x_i \}$.
We show that the total charge to the fractional items is sufficient to pay for the subsidy.
We omit the proof since it is a word-by-word copy-paste of the proof of Lemma~\ref{lemma:upper-bound-threshold-rounding}.

\begin{lemma} \label{lemma:upper-bound-threshold-rounding-goods}
    For all $i\in \{1,2,\ldots,n-1\}$, let $p_i = \min\{ x_i, 1-x_i \}$. Then under the threshold rounding we have $\|\bs\|_1 \leq \sum_{i=1}^{n-1} p_i$.
\end{lemma}

Finally, we combine Lemma~\ref{lemma:upper-bound-down-rounding} and~\ref{lemma:upper-bound-threshold-rounding-goods} to prove the following.

\begin{theorem}\label{theorem:n/4-subsidy-goods}
    Given the fractional allocation $\bX^0$ with $n-1$ fractional items returned by Algorithm~\ref{alg:MKA-goods}, there exists a rounding scheme that returns an integral PROPS allocation $\bX$ with total subsidy at most $n/4$.
\end{theorem}
\begin{proof}
    By Lemma~\ref{lemma:upper-bound-down-rounding}, there exists a sequence of indices $1\leq j_1 < i_2 < j_2 < \cdots < i_z < j_z \leq n-1$ such that the total subsidy required by the down rounding is at most
    \begin{equation}\label{eq:up-rounding-goods}
        x_{j_1} + (x_{j_2} - x_{i_2}) + \cdots + (x_{j_z} - x_{i_z}).
    \end{equation}
    
    Next, we apply different upper bounds for $p_i$ as follows.
    \begin{itemize}
        \item For each $k\in \{i_2, \ldots, i_z\}$, we use $p_i \leq x_{i}$;
        \item For each $k\in \{j_1, \ldots, j_z\}$, we use $p_i \leq 1 - x_i$;
        \item For any other $i$, we use $p_i \leq 0.5$.
    \end{itemize}
    
    Applying the above, the total subsidy required by threshold rounding can be upper bounded by
    \begin{equation}\label{eq:threshold-rounding-goods}
         (1 - x_{j_1}) + (x_{i_2} + 1 - x_{j_2}) + \cdots + (x_{i_z} + 1 - x_{j_z}) + (n-2z)) \cdot \frac{1}{2}.
    \end{equation}
    
    Summing the above two bounds~\eqref{eq:up-rounding-goods} and~\eqref{eq:threshold-rounding-goods}, we have the total subsidy required by the two rounding schemes combined is at most $z + (n-2z)\cdot \frac{1}{2} = \frac{n}{2}$.
    Therefore, at least one of the two rounding schemes requires a total subsidy of at most $n/4$, which proves Theorem~\ref{theorem:n/4-subsidy-goods}.
\end{proof}

\subsection{Upper Bound on the Total Subsidy: Weighted Case}

In this section, we consider the case when different agents may have different weights for the allocation of goods.
The algorithm and analysis we show in this section are similar to those in Section~\ref{ssec:weighted_general} (for the allocation of chores): we first show that the fractional Bid-and-Take Algorithm computes a WPROP allocation, based on which we compute a WPROPS allocation with total subsidy at most $(n-1)/2$.

\subsubsection{Fractional Bid-and-Take Algorithm}

We use the same notation defined in Section~\ref{ssec:weighted_general} for representing a fractional allocation, e.g., we use $x_{ie}\in [0,1]$ to denote the fraction of item $e$ that is allocated to agent $i$.

\paragraph{The Algorithm.}
We allocate the items one-by-one in a continuous manner following an arbitrarily fixed ordering $e_1,e_2,\ldots,e_m$ of the items.
Initially all agents are active.
For each item $e_j\in M$, we continuously allocate $e_j$ to the active agent $i$ with the maximum $\frac{v_i(e_j)}{v_i(M)}$, until either $e_j$ is fully allocated or $v_i(X_i) = \WPROP_i$. 
If $v_i(X_i) = \WPROP_i$, we inactive agent $i$. 
The algorithm terminates when all items are fully allocated. The steps of the full algorithm are summarized in Algorithm~\ref{alg:FBTA-goods}.
\begin{algorithm}[htbp]
    \caption{Fractional Bid and Take Algorithm}
    \label{alg:FBTA-goods}
    \KwIn{An instance $(M,N,\bw,\mathbf{v})$}
    $X_i \gets \mathbf{0}^m, \forall i\in N$  \qquad \qquad \tcp{current fractional bundle}
    $A\gets N$ \qquad \qquad  \qquad \qquad \tcp{ the set of active agents}
    $\mathbf{z} \gets \mathbf{1}^m$ \qquad \qquad  \qquad \qquad\tcp{remaining fraction of the items}
    $j\gets 1$ \qquad \qquad  \qquad \qquad \> \> \tcp{item to be allocated}
    \While{$j \le m$}{

        Let $i \gets \argmax_{i'\in A} \frac{v_{i'}(e_j)}{v_{i'}(M)}$\;
        \If{$v_i(X_i) +z_j\cdot v_i(e_j)>\WPROP_i$}{
            $x_{ie_{j}} \gets \frac{\WPROP_i-v_i(X_i)}{v_i(e_j)}$\;
            $z_j \gets z_j - x_{ie_{j}}$, $A \gets A\setminus \{i\}$\;
            \If{$|A| = 1$}{
                Allocate all remaining items to the only active agent and go to output\;
            }
        }
        \Else{
            $x_{ie_{j}} \gets z_{j}$, $z_{j} \gets 0$, $j\gets j+1$\;
        }    
        
    }

    \KwOut{A fractional allocation $\bX^0  = (X_1,\dots, X_n)$.}
\end{algorithm}


\begin{lemma}\label{lemma: wprop-goods}
    The output allocation $\bX^0$ is a fractional WPROP allocation. 
\end{lemma}
\begin{proof}
    It suffices to show that the last active agent $n$ receives a bundle of value at least $\WPROP_n$. 
    Since we allocate each item $e$ to the active agent $i$ with the maximum $\frac{v_i(e)}{v_i(M)}$, we have $\frac{v_i(e)}{v_i(M)} \geq \frac{v_j(e)}{v_j(M)}$ for every active $j\neq i$. Therefore we have $\frac{v_j(X_i)}{v_j(M)} \leq \frac{v_i(X_i)}{v_i(M)}$ for any active agents $i,j\in A$. Suppose that $v_n(X_n) < \WPROP_n$ at the end of the algorithm, then we have 
    \begin{equation*}
        1 = \sum_{i\in N} \frac{v_n(X_i)}{v_n(M)}
        \le \sum_{i\in N}\frac{v_i(X_i)}{v_i(M)} < \sum_{i\in N}\frac{\WPROP_i}{v_i(M)} = \sum_{i\in N} w_i = 1,
    \end{equation*}
    which is a contradiction.
\end{proof}

Note that in the above WPROP allocation, the number of fractional items is at most $n-1$, since the number of fractional items can increase only when some agent becomes inactive, and there is at least one active agent when the algorithm terminates.

\subsubsection{Rounding Scheme and Subsidy}

Given $\bX^0$ returned by Algorithm~\ref{alg:FBTA-goods}, for every item $e\in M$, we use $k(e) = \{i\in N: x_{ie} > 0\}$ to denote the set of agents who gets (a fraction of) item $e$.
We call an item $e$ \emph{fractional} if and only if $|k(e)| \geq 2$.
%

\paragraph{Rounding Scheme.}
For any fractional item $e$, let $i^* = \argmax_{i\in k(e)} \{x_{ie}\}$ be the agent who owns the maximum fraction of item $e$.
We round each item $e$ to agent $i^*$, i.e., set $x_{i^*e} = 1$ and $x_{ie} = 0$ for every $i \in k(e)\setminus \{i^*\}$.
We summarize the rounding scheme in Algorithm~\ref{alg:RS-goods}.

\begin{algorithm}[htbp]
    \caption{Rounding Scheme}
    \label{alg:RS-goods}
    \KwIn{A fractional allocation $\bX^0$.}
    \For{$j = 1,2,\dots, m$}{
        Let $k(e_j) = \{i\in N: x_{ie_j} > 0 \}$ \;
        \If{$|k(e_j)| \geq 2$}{
            Let $i^* \gets \argmax_{i\in k(e_j)} \{x_{ie_j}\}$, and update $x_{i^*e_j} \gets 1$\;
            Update $x_{ie_j} \gets 0$ for every $i \in k(e_j)\setminus \{i^*\}$\;
        }
    }
    Set $s_i \gets \max\{ \WPROP_i-v_i(X_i), 0\}$ for all $i\in N$\;
    \KwOut{An integral allocation $\bX  = (X_1,\dots, X_n)$ with subsidy $\bs$.}
\end{algorithm}

\begin{theorem}\label{theorem:n/2-goods}
     There exists an algorithm that computes WPROP allocations with subsides at most $(n-1)/2$ for the allocation of goods to a group of agents having general additive valuation functions.
\end{theorem}
\begin{proof}
    We denote by $e_{\sigma(i)}$ the last item allocated to agent $i$, where the allocation might be fractional, and assume w.l.o.g. that agent $n$ is the last active agent.
    Note that the set of fractional items is a subset of $F = \{ e_{\sigma(1)}, e_{\sigma(2)},\ldots, e_{\sigma(n-1)} \}$, and $F$ might be a multiset.
    Similar to the analysis in Section~\ref{ssec:weighted_general}, we use a charging argument that charges money to the fractional items.
    We first show that the charged money is sufficient to pay for the subsidy, and then provide an upper bound on the total money we charged.

    We charge an amount of money $p(e) = \frac{|k(e)| - 1}{|k(e)|}$ to each fractional item $e$.
    Observe that the inclusion of item $e$ to $X_i$ incurs an increase in subsidy of all agents other than $i$ by at most $(1-x_{ie}) \cdot v_i(e) \leq 1-x_{ie} \leq 1 - \frac{1}{|k(e)|}$. Thus the charged money is sufficient to pay for the increase in subsidy due to the rounding of item $e$.
    Next we upper bound the total money we charge to the fractional items.

    Consider any fractional item $e$ with $|k(e)| = k \geq 2$.
    Since $e$ is shared by $k$ agents, we know that $e$ is cut $k-1$ times, and thus there exists at least $k-1$ agents $i$ with $e_{\sigma(i)} = e$.
    We re-charge the money $p(e) = \frac{k - 1}{k}$ to these agents, where each agent is charged an amount at most $\frac{1}{k} \leq \frac{1}{2}$.
    Note that over all fractional items, each agent is charged at most once (by item $e_{\sigma(i)}$) and agent $n$ is not charged, which implies that the total money is at most $(n-1)/2$.

\end{proof}

\end{document}